\documentclass[onecolumn,a4paper]{IEEEtran}
\usepackage[utf8]{inputenc} 
\usepackage[T1]{fontenc}
\usepackage{url}              % provides \url{...}
\usepackage{cite}             % improves presentation of citations

\usepackage[cmex10]{amsmath}  % Use the [cmex10] option to ensure complicance
\usepackage{mleftright}       % fix to wrong spacing of \left-,
\mleftright                   % \middle- \right-commands 

\usepackage{graphicx}         % provides \includegraphics{...} to
\usepackage{booktabs}         % fixes poor spacing in tables and
\usepackage{algorithm}
\usepackage{algpseudocode}
\algtext*{EndFor}
\usepackage{units}
\usepackage{amssymb,amsthm}
\usepackage{dsfont}
\usepackage{thmtools,xcolor}

\newtheorem{proposition}{Proposition}
\newtheorem{lemma}{Lemma}
\newtheorem{theorem}{Theorem}
\newtheorem{corollary}{Corollary}

\theoremstyle{definition}
\newtheorem{example}{Example}
\newtheorem{excont}{Example}
\newtheorem{definition}{Definition}

\newcommand{\E}{\textit{e}}

\newcommand{\peterg}[1]{{\color{blue}  #1}}

\newcommand{\commentout}[1]{}

\DeclareMathOperator*{\argmin}{arg\,min}

\begin{document}

\title{Reverse Information Projections and Optimal \texttt{E}-statistics} 

%%%%%%
\author{%
Tyron Lardy, Peter Gr\"unwald and Peter Harremo{\"e}s, \IEEEmembership{Member, IEEE
}%
\thanks{T. Lardy and P. Gr\"unwald are affiliated with CWI, Amsterdam, and Leiden University, The Netherlands. P. Harremo{\"e}s is at Copenhagen Business College, Copenhagen, Denmark. CWI is the national research institute for mathematics and computer science in the Netherlands.}%
\thanks{A five-page abstract of this paper, containing a subset of the theorems but no proofs, was presented at ISIT 2023, Taipei.}%
\thanks{This work has been submitted to the IEEE for possible publication. Copyright may be transferred without notice, after which this version may no longer be accessible.}
}

\maketitle

\begin{abstract}
    Information projections have found important applications in probability theory, statistics, and related areas.
    In the field of hypothesis testing in particular, the reverse information projection (RIPr) has 
    recently been shown to lead to growth-rate optimal (GRO) \E-statistics for testing simple alternatives against composite null hypotheses.
    However, the RIPr as well as the GRO criterion are undefined whenever the infimum information divergence between the null and alternative is infinite. 
    % Tyron review rewritten
    % We show that in such scenarios there often still exists an element in the alternative that is `closest' to the null: the universal reverse information projection. 
    %The universal reverse information projection and its non-universal counterpart coincide whenever information divergence is finite. 
    %Furthermore, the universal RIPr is shown to lead to optimal e-statistics in a sense that is a novel, but natural, extension of the GRO criterion.
    %We also give conditions under which the universal RIPr is a strict sub-probability measure, as well as conditions under which an approximation of the universal RIPr leads to approximate e-statistics.
    We show that in such scenarios, under some assumptions, there still exists a measure in the null that is closest to the alternative in a specific sense. 
    Whenever the information divergence is finite, this measure coincides with the usual RIPr.
    It therefore gives a natural extension of the RIPr to certain cases where the latter was previously not defined. 
    This extended notion of the RIPr is shown to lead to optimal \E-statistics in a sense that is a novel, but natural, extension of the GRO criterion.
    We also give conditions under which the (extension of the) RIPr is a strict sub-probability measure, as well as conditions under which an approximation of the RIPr leads to approximate \E-statistics.
    For this case we provide tight relations between the corresponding approximation rates. 
\end{abstract}
\begin{IEEEkeywords}
Reverse Information Projections, Description Gain, Hypothesis Testing, \texttt{E}-variables.
\end{IEEEkeywords}

\section{Introduction}
\IEEEPARstart{W}{e} write $D(\nu\|\lambda )$ for the information divergence (Kullback-Leibler divergence, \cite{Kullback1951,Csiszar1963,Liese1987}) 
between two finite measures $\nu$ and $\lambda$ given by
\[D(\nu\|\lambda)=\begin{cases} \displaystyle \int_\Omega \ln\left(\frac{\mathrm{d}\nu}{\mathrm{d}\lambda} \right)\, \mathrm{d}\nu - (\nu(\Omega)-\lambda(\Omega)),  & \text{if } \nu\ll \lambda ; \\ \infty,& \text{else.}\end{cases}\] 
For probability measures the interpretation of $D(\nu\|\lambda)$ is that it measures how much we gain by coding 
according to $\nu$ rather than coding according to $\lambda$ if data are distributed according to $\nu$. 
Many problems in probability theory and statistics, such as conditioning and maximum likelihood estimation, 
can be cast as minimization in either or both arguments of the information divergence.
In particular, this is the case within the recently established and now flourishing  theory of hypothesis testing based on \E-statistics that allows for optional continuation of experiments (see Section~\ref{subsec:Estat})\cite{grunwald2024,ramdas2022savi,vovk2021values,shafer2021testing,henzi2022valid}.
That is, a duality
has been established between optimal \E-statistics for testing a simple 
alternative $P$ against a composite null hypothesis $\mathcal{C}$ and reverse information projections~\cite{grunwald2024}.
Here, the reverse information projection (RIPr) of $P$ on $\mathcal{C}$ is --- if it exists --- a unique measure $\hat Q$ such 
that every sequence $(Q_n)_{n\in\mathbb{N}}$ in $\mathcal{C}$ with $D(P\|Q_n)\to\inf_{Q\in\mathcal{C}} D(P\|Q)$ converges to $\hat Q$ in a particular norm~\cite{li1999Estimation,csiszar2003}.
Li \cite{li1999Estimation} showed that whenever $\mathcal{C}$ is convex and $D(P\|\mathcal{C}):=\inf_{Q\in\mathcal{C}} D(P\|Q)<\infty$, the RIPr $\hat Q$ exists and the likelihood ratio between $P$ and $\hat Q$ is an \E-statistic (this result is restated as Theorem~\ref{thm:li} below). Gr\"unwald et al. \cite{grunwald2024} showed (restated as Theorem~\ref{thm:GHKdual} below) that it is even the optimal \E-statistic for testing $P$ against $\mathcal{C}$.
However, it is clear that the RIPr cannot be defined in this way if the information divergence between $P$ 
and $\mathcal{C}$ is infinite, i.e.\ $D(P\|\mathcal{C})=\infty$. 
This leaves a void in the theory of optimality of \E-statistics. 
In this work we remedy this by realizing that even if all measures in $\mathcal{C}$ are infinitely 
worse than $P$ at describing data distributed according to $P$ itself, there can still be a measure that performs best relative to the elements of $\mathcal{C}$.
To find such a measure, we consider the {\em description gain} \cite{Topsoe2007a} given by 
\begin{equation}\label{eq:descr_gain}
    D(P\|Q\rightsquigarrow Q')=\int_\Omega \ln\left(\frac{\mathrm{d}Q'}{\mathrm{d}Q} \right)\, \mathrm{d} P- (Q'(\Omega)-Q(\Omega))
\end{equation}
whenever this integral is well-defined. If the quantities involved are finite then the description gain reduces to
\begin{equation}\label{eq:divdif}
    D(P\|Q\rightsquigarrow Q')=D(P\|Q)-D(P\|Q').
\end{equation}
In analogy to the interpretation of information divergence for coding, the description gain measures how much we gain by 
coding according to $Q'$ rather than $Q$ if data are distributed according to $P$. 
Furthermore denote 
\[
D(P\|Q\rightsquigarrow \mathcal{C}):=\sup_{Q'\in \mathcal{C}} D(P\|Q\rightsquigarrow Q'),
\]
where undefined values are counted as $-\infty$ when taking the supremum.
If there exists at least one $Q^*\in \mathcal{C}$ such that $P\ll Q^*$, then $D(P\|Q\rightsquigarrow \mathcal{C})$ is a well-defined number in $[0,\infty]$ for any $Q\in \mathcal{C}$.
This quantity should be seen as the maximum description gain one can get by switching from $Q$ to any other measure in $\mathcal{C}$.
Intuitively, if there is a best descriptor in $\mathcal{C}$, nothing can be gained by switching away from it.
Indeed, in Proposition~\ref{prop:equvalence} we show that $\inf_{Q\in \mathcal{C}} D(P\|Q\rightsquigarrow \mathcal{C})$ is finite if and only if it is equal to zero.

\subsection{Contents and Overview}
Below, in Section~\ref{sec:bekgrond}, we start by giving an overview of existing results on both the reverse information projection and $\E$-statistics, which we define and briefly motivate, and the growth-rate optimality (GRO) criterion, a natural replacement of statistical power within the context of $\E$-value based hypothesis testing. 
%
%Furthermore, in Theorem \ref{thm:InfoGain}, we show that
Section~\ref{sec:universal_ripr}  states Theorem \ref{thm:InfoGain}, our first central result. It shows that 
% Universal removed
%
--- under very mild conditions --- there exists a 
unique measure $\hat Q$ such that every sequence $(Q_n)_{n\in \mathbb{N}}$ in $\mathcal{C}$ with 
\[
D(P\|Q_n\rightsquigarrow \mathcal{C})\to 0 %\inf_{Q\in \mathcal{C}} D(P\|Q\rightsquigarrow \mathcal{C})
\]
converges to $\hat Q$ 
in a specific metric which we define. Thus, Theorem \ref{thm:InfoGain}  may be viewed as 
a generalization of Li's result stated below as Theorem~\ref{thm:li}. 
% Tyron review universal removed (twice)
We refer to $\hat Q$ as the RIPr, as it coincides with the original notion of the RIPr whenever the information divergence is finite. 
The remainder of Section~\ref{sec:universal_ripr} provides further discussion of this result, as well as an example showing that our extended notion of the RIPr can be well-defined whereas the RIPr was previously undefined.
In the specific case that all initial measures are probability measures, both Li's original result and ours leave open the possibility that $\hat{Q}$ may be a strict sub-probability measure, integrating to less than 1.
In Sub-Section~\ref{sec:strictlyspeaking} we give a further example showing that this can indeed be the case, and we provide, via  Theorem~\ref{thm:aftagende}, a condition under which $\hat{Q}$ is guaranteed to be a standard (integrating to $1$) probability measure.  
% Tyron review universal removed
Sub-Section~\ref{sec:algo} then extends the greedy algorithm of \cite{li1999mixture} and \cite{brinda2018adaptive}
for approximating the RIPr in settings where $D(P\|\mathcal{C})<\infty$ to settings where the information divergence might be infinite. 

In Section~\ref{sec:ordering} we turn to \E-statistics. 
% Tyron review universal removed
It contains our second central result, Theorem \ref{thm:LR_optim}, which shows that whenever our extended notion of the RIPr $\hat Q$ exists, the  likelihood ratio of $P$ and $\hat Q$ is an optimal \E-statistic 
according to the criterion of Definition~\ref{def:ordering}, which can be seen as a strict generalization of GRO, the standard optimality criterion for \E-statistics. 
As such, this result may be viewed as a generalization of a result of
Gr\"unwald et al. \cite{grunwald2024} stated below as Theorem~\ref{thm:GHKdual}. 
% Tyron review universal removed
After illustrating the result by an example, Sub-Section~\ref{sec:approximation} provides another technical result, Theorem~\ref{thm:approx}, which relates approximations in terms of information gain, to approximations in terms of \E-statisticity: conditions are given under which a sequence $Q_1, Q_2, \ldots$ converging to $\hat{Q}$ in terms of information gain at a certain rate  also satisfies that the likelihood ratio between $P$ and $Q_1, Q_2, \ldots$ converges to an \E-statistic, and tight bounds on the corresponding rates are given. 
After a discussion of related work, the paper ends with a summary and ideas for future work in Section~\ref{sec:summary}. All proofs are delegated to Appendix~\ref{app:proofs}.
Appendix~\ref{app:supp} provides a general method for constructing RIPrs that are strict sub-probability measures. 
Finally, Appendix~\ref{ap:convexity} provides a discussion on the assumption of convexity that we will make throughout.

\section{Background}\label{sec:bekgrond}
\subsection{Preliminaries}
We work with a measurable space $(\Omega, \mathcal{F})$ and, unless specified otherwise, all measures will be defined on this space. 
Throughout, $P$ will denote a finite measure and $\mathcal{C}$ a set of finite measures, 
such that $P$ and all $Q\in\mathcal{C}$ have densities w.r.t. a common $\sigma$-finite measure $\mu$. 
These densities will be denoted with lowercase, i.e.\ $p$ and $q$ respectively. 
We will assume throughout that $\mathcal{C}$ is convex, i.e.\ closed under finite mixtures. In Section~\ref{sec:convexity} and in more detail in Appendix~\ref{ap:convexity} we discuss how our results would be affected if we were to adopt stronger notions of convexity like $\sigma$-convexity (closed under countable mixtures), or Choquet-convexity (closed under arbitrary mixtures). 
Furthermore, we assume that there exists at least one $Q^*\in \mathcal{C}$ such that $P\ll Q^*$.
This assumption is needed to ensure that $D(P\|Q\rightsquigarrow \mathcal{C})$ is a well-defined number in $[0,\infty]$ for any $Q\in \mathcal{C}$.

% Tyron review removed
% On the other hand, it aligns with our philosophy when we turn to hypothesis testing, in which case $P$ and all $Q\in\mathcal{C}$ will be probability measures and serve as the alternative and null hypothesis respectively.
% We will consider $P$ mostly as a tool to gather evidence against $\mathcal{C}$, so that it does not make sense to consider the case in which $P$ puts mass on events that cannot occur according to the null, as the null hypothesis can be discredited in such scenarios regardless of how much mass $P$ puts on the event.
%

\subsection{The Reverse Information Projection}\label{sec:RIPr}
As mentioned briefly above, the reverse information projection is the result of minimizing the information divergence between $P$ and $\mathcal{C}$.
If $\mathcal{C}$ is an exponential family, this problem is well understood~\cite{csiszar2003}, 
but we focus here on the case that $\mathcal{C}$ is a general convex set.
In this setting, the following theorem establishes existence and uniqueness of a limiting object for 
any sequence $(Q_n)_{n\in\mathbb{N}}$ in $\mathcal{C}$ such that $D(P\|Q_n)\to D(P\|\mathcal{C})$ whenever the latter is finite.
This limit (i.e. $\hat Q$ in the following) is called the reverse information projection of $P$ on $\mathcal{C}$.

\begin{theorem}[Li~\cite{li1999Estimation}, Definition 4.2 and Theorem 4.3]\label{thm:li} 
If $P$ and all $Q\in\mathcal{C}$ are probability measures such that $D(P\| \mathcal{C})<\infty$, 
then there exists a unique (potentially sub-) probability measure $\hat Q$ such that: 
\begin{enumerate}
    \item\label{item:ripr_conv} We have that $\ln q_n \rightarrow \ln \hat q$ in $L_1(P)$ for all sequences $(Q_n)_{n\in\mathbb{N}}$ in $\mathcal{C}$ such that $\displaystyle\lim_{n\rightarrow \infty} D(P\|Q_n)=D(P\|\mathcal{C})$. 
    \item  $\int_\Omega \ln \frac{\mathrm{d}P}{\mathrm{d}\hat Q}\, \mathrm{d}P=D(P\| \mathcal{C})$,
    \item $\int_{\Omega}\frac{\mathrm{d}P}{\mathrm{d}\hat Q}\,\mathrm{d}Q\leq 1$ for all $Q\in\mathcal{C}$. 
    \end{enumerate}
\end{theorem}

\subsection{E-statistics and Growth Rate Optimality}\label{subsec:Estat}
The \E-value has recently emerged as a popular alternative to the {\it p}-value for hypothesis testing~\cite{ramdas2022savi,grunwald2022anytime,henzi2022valid}.
Unlike the {\it p}-value, it is eminently suited for testing under  optional continuation --- and more generally, when the rule for stopping or continuing to analyze an additional batch of data is not under control of the data analyst, and may even be unknown or unknowable.
It can be thought of as a measure of statistical evidence that is intimately linked with numerous ideas, such as likelihood ratios, test martingales~\cite{ville1939} and tests of randomness~\cite{levin1976uniform}.
Formally, an \E-value is defined as the value taken by an \E-statistic, which is defined as a random variable $E:\Omega \rightarrow [0,\infty]$ that satisfies $\int_\Omega E \,\mathrm{d}Q \leq 1$ for all $Q\in \mathcal{C}$~\cite{vovk2021values}.
The set of all \E-statistics is denoted as $\mathcal{E}_\mathcal{C}$.
Large \E-values constitute evidence against $\mathcal{C}$ as null hypothesis, so that the null can be rejected when the computed \E-value exceeds a certain threshold.
For example, the test that rejects the null hypothesis when $E\geq \nicefrac{1}{\alpha}$ has a type-I error guarantee of $\alpha$ by a simple application of Markov’s inequality:
$Q(E\geq \nicefrac{1}{\alpha})\leq \alpha \int_\Omega E \,\mathrm{d}Q\leq \alpha$.
For all further details, as well as an extensive introduction to the concept, and how it relates to optional stopping and continuation, we refer to \cite{grunwald2024} and the overview paper \cite{ramdas2022savi}. 

In general, the set $\mathcal{E}_\mathcal{C}$ of \E-statistics is quite large, and the above does not tell us \emph{which} \E-statistic to pick.
This question was studied in~\cite{grunwald2024} and a log-optimality criterion coined GRO ({\em Growth-Rate Optimality\/})  was introduced for the case that the interest is in gaining as much evidence as possible relative to an alternative hypothesis given by a single probability measure $P$. 
GRO is a natural replacement of statistical power, which cannot meaningfully be used in an optional stopping/continuation context. 
This criterion can be traced back to the information-theoretic  Kelly betting criterion in~\cite{kelly1956ANI} and is further discussed at length by~\cite{shafer2021testing,ramdas2022savi,grunwald2024}, to which we refer for more discussion. 

\begin{definition}\label{def:GHKgro}
If it exists, an \E-statistic $\hat{E}\in \mathcal{E}_\mathcal{C}$ is Growth-Rate Optimal (GRO) if it achieves
\[\int_\Omega \ln \hat{E} \, \mathrm{d}P=\sup_{E\in \mathcal{E}_\mathcal{C}} \int_\Omega \ln E \, \mathrm{d}P.\] 
\end{definition}

The following theorem establishes a duality between GRO \E-statistics and reverse information projections.
For a limited set of testing problems, it states that GRO \E-statistics exist and are uniquely given by likelihood ratios.

\begin{theorem}[Gr\"unwald et al.~\cite{grunwald2024}, Theorem 1]\label{thm:GHKdual}
If $P$ and all $Q\in\mathcal{C}$ are probability measures such that $D(P\| \mathcal{C})<\infty$, $p(\omega)>0$ for all $\omega\in \Omega$, and
$\hat Q$ is the RIPr of $P$ on $\mathcal{C}$, then $\hat E=\frac{\mathrm{d}P}{\mathrm{d} \hat Q}$ is GRO with rate equal to $D(P\|\mathcal{C})$, i.e.
\[\sup_{E\in \mathcal{E}_\mathcal{C}} \int_\Omega \ln E \, \mathrm{d}P= \int_\Omega \ln \hat E \, \mathrm{d}P= D(P\|\mathcal{C}).\] 
Furthermore, for any GRO \E-statistic $\tilde E$, we have that $\tilde E=  \hat E$ holds $P$-almost surely.
\end{theorem}

%Tyron review universal removed
\section{The Reverse Information Projection}\label{sec:universal_ripr} 
In this section, we state a result analogous to Theorem~\ref{thm:li} in a more general setting.
Rather than convergence of the logarithm of densities in $L_1(P)$, we consider convergence with respect to a 
% Tyron review added
different 
metric on the set of measurable positive functions, i.e.\ $\mathcal{M}\left(\Omega,\mathbb{R}_{>0}\right)=\{f:\Omega \to \mathbb{R}_{>0}: f \text{ measurable}\}$. 
For $f,f' \in \mathcal{M}\left(\Omega,\mathbb{R}_{>0}\right)$ we define
\begin{equation}\label{eq:jens_div}
    m_P^2(f, f'):=\frac12 \int_\Omega\ln\left(\frac{\overline f}{f}\right)+\ln\left(\frac{\overline f}{f'}\right)\, \mathrm{d}P,
\end{equation}
where $\overline{f}:=\nicefrac{(f+f')}{2}$.
This is a divergence that can be thought of as the averaged Bregman divergence associated with the convex function $\gamma(x)=x-1-\ln(x)$.
In particular, this means that for $Q,Q'\in \mathcal{C}$ such that $P\ll Q$ and $P\ll Q'$, we have that 
\begin{equation}\label{eq:SymGain}
    m_P^2(q,q')=\frac12 D(P\|Q\rightsquigarrow \bar{Q})+ \frac12 D(P\|Q'\rightsquigarrow \bar{Q}).
\end{equation}
In~\cite{chen2008metrics}, averaged Bregman divergences are studied in detail for general $\gamma$, and they show that the function 
\[m^2_\gamma(x,y)=\frac12 \gamma(x) + \frac12 \gamma(y) - \gamma\left(\frac{x+y}{2}\right)\] 
is the square of a metric if and only if $\ln\left(\gamma''(x)\right)''\geq 0$. 
In our case, $\ln(\gamma''(x))''=2x^{-2}$, so this result holds.
This can be used together with an application of the  Minkowski inequality to show that the triangle inequality holds for the square root of the divergence \eqref{eq:jens_div}, i.e.\ $m_P$, on $\mathcal{M}\left(\Omega,\mathbb{R}_{>0}\right)$. 
It should also be clear that for $f,g\in \mathcal{M}\left(\Omega,\mathbb{R}_{>0}\right)$ if $f=g$ everywhere, then $m_P(f,g)=0$.
Conversely $m_P(f,g)=0$ only implies that $P(f\neq g)=0$. 
This prevents us from calling $m_P$ a metric on $\mathcal{M}\left(\Omega,\mathbb{R}_{>0}\right)$, and we therefore define, analogous to $\mathcal{L}^p$ and $L^p$ spaces, $M\left(\Omega, \mathbb{R}_{>0}\right)$ as the set of equivalence classes of $\mathcal{M}\left(\Omega,\mathbb{R}_{>0}\right)$ under the relation `$\sim$' given by $f\sim g \Leftrightarrow P(f\neq g)=0$.
By the discussion above, $m_P$ properly defines a metric on $M\left(\Omega, \mathbb{R}_{>0}\right)$.
In the following we will often ignore this technicality and simply act as if $m_P$ defines a metric on $\mathcal{M}\left(\Omega, \mathbb{R}_{>0}\right)$, since we are not interested in what happens on null sets of $P$.

% Tyron review added
Considering convergence with respect to $m_P$ will be useful for our analyses in the following.
In particular, we will exploit on numerous occasions that $m_P$ can be interpreted as a symmetrized version of the description gain, as described in Equation (\ref{eq:SymGain}).
However, other than mathematical convenience, there is no fundamental difference between considering convergence with respect to $m_P$ and convergence of the logarithms in $L_1(P)$, as considered in Theorem~\ref{thm:li}.
Indeed, Lemma~\ref{lem:aux2} in Appendix~\ref{app:proofs} shows that the two types of convergence are equivalent.
It is also this result from which the following proposition follows.

\begin{proposition}\label{prop:complete}
The metric space $(M\left(\Omega,\mathbb{R}_{>0}\right),m_P)$ is complete.
\end{proposition}

Everything is now in place to state the main result.

\begin{theorem}\label{thm:InfoGain}
If $\inf_{Q\in\mathcal{C}} D(P\|Q\rightsquigarrow \mathcal{C})<\infty$,
then there exists a measure $\hat Q$ that satisfies the following for every sequence  $(Q_n)_{n\in \mathbb{N}}$ in $\mathcal{C}$ s.t. $D(P\|Q_n\rightsquigarrow \mathcal{C})\to \inf_{Q\in\mathcal{C}} D(P\|Q\rightsquigarrow \mathcal{C})$ as $n\to\infty$:
\begin{enumerate}%[label={(\arabic*)}]
    \item\label{item:conv_mpny} $q_n\rightarrow \hat q$ in $m_P$.
    \item\label{item:conv_lrny} If  $P'$ is a measure such that $\left|\inf_{Q\in\mathcal{C}}D(P\|Q\rightsquigarrow P')\right|<\infty$, then 
    %$\displaystyle 
    $$
    \int_\Omega \ln \frac{\mathrm{d}P'}{\mathrm{d}\hat Q} \, \mathrm{d}P = \lim_{n\to \infty} \int_\Omega \ln \frac{\mathrm{d}P'}{\mathrm{d}Q_n} \, \mathrm{d}P.
    $$
    \item\label{item:conv_evalny} For any $Q\in \mathcal{C}$, 
    %$\displaystyle
    $$\int_\Omega \frac{\mathrm{d}P}{\mathrm{d}\hat Q} \, \mathrm{d}Q \leq P(\Omega)+Q(\Omega)-\liminf_{n\to \infty}Q_n(\Omega).$$ 
    %for any $Q\in \mathcal{C}$.
\end{enumerate}
\end{theorem}

Theorem~\ref{thm:li} is a special case of Theorem~\ref{thm:InfoGain} 
when $P$ and all $Q\in \mathcal{C}$ are probability measures such that $D(P\|\mathcal{C})<\infty$.
This follows because Equation \eqref{eq:divdif} implies that minimizing $D(P\|Q\rightsquigarrow Q')$ over $Q$ is equivalent to minimizing $D(P\|Q)$ and because convergence of the densities in $m_P$ is equivalent to convergence of the logarithms in $L_1(P)$ by Lemma~\ref{lem:aux2} in Appendix~\ref{app:li_proof}.
% Tyron review universal removed
We therefore refer to $\hat Q$ as the \emph{reverse information projection} of $P$ on $\mathcal{C}$, thereby extending the definition of the latter (we refrain from the term `generalized RIPr', because it has already been used for the RIPr whenever it is not attained by an element of $\mathcal{C}$~\cite{csiszar2003} or when the log score is replaced by another loss function \cite{GrunwaldM20}). 
However, the density of the measure $\hat Q$ is only unique as an element of $M(\Omega,\mathbb{R}_{>0})$, since convergence of the densities holds in $m_P$. 
In the current work this causes no ambiguity, so that we simply refer to it as `the' RIPr. 

Note that Theorem~\ref{thm:InfoGain} implies that if there exists a $Q\in \mathcal{C}$ with $D(P\|Q\rightsquigarrow \mathcal{C})=0$, then $Q$ is the RIPr of $P$ on $\mathcal{C}$.
This matches with the intuition that the maximum gain we can get from switching away from the `best' code in $\mathcal{C}$ should be equal to zero.
% Tyron review edited
The following result establishes this more formally. 

\begin{proposition}\label{prop:equvalence}
    The following conditions are equivalent:
    \begin{enumerate}%[label={(\arabic*)}]
        \item There exists a measure $P'$  such that $D(P\|P'\rightsquigarrow \mathcal{C})$ is finite.
        \item There exists a measure $Q$  in $\mathcal{C}$ such that $D(P\|Q\rightsquigarrow \mathcal{C})$ is finite.
        \item There exists a sequence of measures $Q_n\in \mathcal{C}$ such that $D(P\|Q_n\rightsquigarrow \mathcal{C})\to 0$ for $n\to\infty$.
    \end{enumerate}
    % Tyron review added
    Consequently, whenever $\inf_{Q\in \mathcal{C}}D(P\|Q\rightsquigarrow \mathcal{C})<\infty$, it must actually be equal to zero.
\end{proposition}

To show that the reverse information projection exists, it is therefore enough to prove that one of these equivalent conditions holds.
Which condition is easiest to check will depend on the specific setting, as exemplified by the following propositions. 

\begin{proposition}\label{prop:convhull}
    If $\mathcal{C}$ is the convex hull of finitely many distributions, i.e. $\mathcal{C}=\mathrm{conv}(\{Q_1,\dots,Q_n\})$, then for any probability measure $P$ with $P\ll Q_i$ for at least one $i$, it holds that $D(P\| \frac{1}{n}\sum Q_i \rightsquigarrow \mathcal{C})<\infty$.
\end{proposition}

\begin{example}\label{ex:gauss_cauchy}
Let $\mathcal{C}$ be a singleton whose single element $Q$ is given by the standard Gaussian and let $P$ be the standard Cauchy distribution.
Since the Cauchy distribution is exponentially heavier-tailed than the Gaussian, 
we have that $D(P\|\mathcal{C})=\infty$.
However, since both distributions have full support, it follows that
\[
D(P\|Q\rightsquigarrow \mathcal{C})=D(P\|Q\rightsquigarrow Q)=0.
\] 
% Tyron review universal removed
By Theorem~\ref{thm:InfoGain}~(1), $Q$ is therefore the reverse information projection of $P$ on $\mathcal{C}$. 

This example can be extended to composite $\mathcal{C}$ by considering all mixtures of the Gaussian distributions $\mathcal{N}(-1,1)$ and $\mathcal{N}(1,1)$ with mean $\pm1$  and variance $1$.
Proposition \ref{prop:convhull} guarantees the existence of a reverse information projection although the information divergence is still infinite because a Cauchy distribution is more heavy tailed than any finite mixture of Gaussian distributions. 
% Tyron review universal removed
Symmetry implies that the reverse information projection must be equal to the uniform mixture of $\mathcal{N}(-1,1)$ and $\mathcal{N}(1,1)$,
which coincides with the result one would intuitively expect. 
 \end{example}

\begin{proposition}\label{prop:minimax}
    Assume that $\mathcal{C}$ is a convex set of probability measures that has finite minimax regret and with 
    normalized maximum likelihood 
    distribution $Q^*\in\mathcal{C}$. 
    Then for any probability measure $P$ that 
    is absolutely continuous with respect to $Q^*$, it holds that $D(P\|Q^*\rightsquigarrow \mathcal{C})<\infty$.
\end{proposition}
For an extensive discussion on minimax regret in the present coding context, as well as the normalized maximum likelihood distribution (also known as {\em Shtarkov\/} distribution), see e.g. \cite{Grunwald2009} or \cite{Erven2014}.
In short, the minimax regret is defined as $\inf_{Q\in \mathcal{C}} \sup_{Q'\in \mathcal{C},\omega\in\Omega} \ln \nicefrac{q'(\omega)}{q(\omega)}$.
This quantity is known to be finite if and only if the normalized maximum likelihood distribution $Q^*$, defined as $q^*(\omega)=\nicefrac{\sup_{Q\in \mathcal{C}} q(\omega)}{\left(\int_\Omega \sup_{Q\in \mathcal{C}} q \, \mathrm{d}\mu\right)}$, is well-defined.
One-dimensional exponential families with finite minimax regret have been classified in \cite{Grunwald2009}.

\subsection{Strict sub-probability measure}
\label{sec:strictlyspeaking}
We return now to the familiar setting where $P$ is a probability measure and $\mathcal{C}$ a convex set of probability measures.
It is easy to verify that the RIPr $\hat Q$ of $P$ on $\mathcal{C}$ is then a sub-probability measure.
This follows because we know that there exists a sequence $(Q_n)_{n\in \mathbb{N}}$ in $\mathcal{C}$ 
such that $q_n$ converges point-wise $P$-a.s. to $\hat q$ and Fatou's Lemma tells us
\begin{align}
%for two columm version: \begin{multline*}
    \int_\Omega \hat q\,\mathrm{d}\mu =\int_\Omega \liminf_{n\to\infty} q_n\,\mathrm{d}\mu
    \leq \liminf_{n\to \infty} \int_\Omega q_n \, \mathrm{d}\mu=1.
%\end{multline*}
\end{align}
It is not clear a priori whether this can ever be a strict inequality.
For example, if the sample space is finite, the set of probability measures is compact, so the 
limit of any sequence of probability measures (i.e.\ the reverse information projection) will also be a probability measure.
The following example illustrates that this is not always the case for infinite sample spaces, 
and it can in fact already go wrong for a countable sample space with $D(P\|\mathcal{C})<\infty$.

\begin{example}
Let $\Omega=\mathbb{N}$ and $\mathcal{F}=2^\mathbb{N}$. 
Furthermore, let $P$ denote the probability measure $\delta_1$ concentrated in the point $i=1$ and $\mathcal{C}$ the set of distributions $Q$ satisfying
\[\sum_{i=1}^\infty \frac1i q(i)=\frac12.\] 
This set is defined by a linear constraint, so that $\mathcal{C}$ is convex, and for any $Q\in \mathcal{C}$, we have
\[q(1)+\sum_{i=2}^\infty \frac1i q(i) =\sum_{i=1}^\infty \frac1i q(i)=\frac12,\]
implying that $q(1)\leq \nicefrac{1}{2}$. 
It follows that $D(P\|Q)=-\ln(q(1))\geq \ln(2)$. 
The sequence $Q_n=\frac{n-2}{2n-2}\delta_1+\frac{n}{2n-2}\delta_n$ satisfies $Q_n\in\mathcal{C}$ and 
\[D(P\|Q_n)=\ln\frac{2n-2}{n-2} \to \ln(2).\]
Consequently, it must hold that $D(P\|\mathcal{C})=\ln(2)$.
% Tyron review universal removed
The sequence $Q_n$ converges to the strict sub-probability measure $(\nicefrac12) \delta_1$, which must therefore be the RIPr of $P$ on $\mathcal{C}$. 
\end{example}

A more general example, which can be seen as a template to create such situations, is given in Appendix~\ref{app:supp}.
The common theme is that $\mathcal{C}$ is defined using only constraints of the form $\sum_i f_1(i)q(i)=c$, where $f_1$ is some positive function such that $\lim_{n\to \infty} f_1(n)=0$.
Since $\mathcal{C}$ only contains probability measures, there is the additional constraint that $\sum_i f_0(i)q(i) =1$, where $f_0$ denotes the constant function $f_0\equiv 1$.
This function $f_0$ dominates all other constraints $f_1$ in the sense that $\lim_{i\to \infty} \nicefrac{f_1(i)}{f_0(i)}=0$, but is itself not dominated by any of the constraints in the same manner.
% Tyron review added sentence & slight rewrite of second sentence
It turns out that this is the precise condition that dictates whether or not a constraint has to be respected when taking point-wise limits of elements in $\mathcal{C}$.
Indeed, as shown in the theorem below, any constraint on $\mathcal{C}$ that is dominated by another constraint in the sense described above cannot be violated by taking point-wise limits.
Therefore, if we add a restriction to $\mathcal{C}$ that dominates the constant function $1$, i.e. that is defined by some function $f_1$ with $\lim_{n\to \infty} f_1(n)=\infty$, then the RIPr cannot be a strict sub-probability measure.

\begin{theorem}\label{thm:aftagende}
Take $\Omega=\mathbb{N},\mathcal{F}=2^\mathbb{N}$, and let $\mathcal{C}$ be a convex set of probability measures.
Suppose that for $f_0,f_1:\mathbb{N}\to \mathbb{R}_{>0}$, we have that $\sum_i f_0(i) q(i)\leq \lambda_0$ and $\sum_i f_1(i) q(i)=\lambda_1$ for all $Q\in \mathcal{C}$. 
If $Q_{n}$ denotes a sequence of measures in $\mathcal{C}$ that converges point-wise to some distribution $Q^{*}$, and $f_{0}$ dominates $f_{1}$ in the sense that
\begin{equation}
    \lim_{i\to \infty} \frac{f_{1}\left(i\right)}{f_{0}\left(i\right)} = 0,
\end{equation}
then 
\begin{equation}
    \sum_{i}f_{1}\left(i\right)\cdot q^{*}\left(i\right)=\lambda_{1}.
\end{equation}
\end{theorem}

\subsection{Greedy Approximation}\label{sec:algo}
% Tyron review universal removed
So far, we have discussed the existence and properties of the RIPr of $P$ on $\mathcal{C}$. 
However, there will be many situations where it is infeasible to compute this exact projection, as it requires solving a complex minimization problem. 
% Tyron review universal removed
For example, if $\mathcal{C}$ is given by the convex hull of some parameterized family of distributions, the reverse information projection might be an arbitrary mixture of elements of this family, and the minimization problem need not be convex in the parameters of the family. 
To this end, Li and Barron~\cite{li1999mixture} propose an iterative greedy algorithm for the case that $\mathcal{C}$ is given by the $\sigma$-convex hull (all countable mixtures, see Appendix~\ref{ap:convexity}) of a parameterized family of distributions, i.e. $\mathcal{C}=\sigma\text{-conv}(\{Q_\theta: \theta\in \Theta\})$, and $D(P\|\mathcal{C})<\infty$. 
The algorithm starts by setting $Q_1 := Q_{\theta_1}$, where $\theta_1$ minimizes $D(P\|Q_{\theta_1})$, and then iteratively defining $Q_k:= (1-\alpha_k)Q_{k-1} + \alpha_k Q_{\theta_k}$, where $\alpha_k = \nicefrac{2}{(k+1)}$\footnote{Li actually proposes to either minimize over $\alpha_k$ or use $\alpha_2=\nicefrac{2}{3}$ and $\alpha_k=\nicefrac{2}{k}$ for $k>2$; the formulation given here is a slight simplification by Brinda~\cite{brinda2018adaptive}.} and $\theta_k$ is chosen to minimize $D(P\|Q_k)$.
It is shown that, if $\sup_{x,\theta_1,\theta_2} \log \nicefrac{q_{\theta_1}(x)}{q_{\theta_2(x)}}$ is bounded, then $D(P\|Q_k)$ converges to $D(P\|\mathcal{C})$ at rate $1/k$.
Later, Brinda~\cite{brinda2018adaptive} showed that the condition that  the likelihood ratio has to be uniformly bounded in $x$ can be relaxed to the condition that (\ref{eq:brindabound}) below is finite. 
In both of these previous works, it is simply assumed that a minimizer in each step exists, though it need not necessarily be unique.
We will do likewise in the following, where we give an adaptation of the algorithm that works when the KL divergence is infinite.  

\begin{algorithm}
% Tyron review universal removed
\caption{Greedy Approximation of the RIPr}\label{alg:ripr} 
\begin{algorithmic}[1]
   \State Fix $Q^*\in \mathcal{C}$ s.t. $|\inf_{\theta \in \Theta} \int_\Omega \log \nicefrac{q^*}{q_\theta}\, \mathrm{d}P|<\infty$ %(This is not needed, since we assume D(P\|Q' \rightsquigarrow Q'') is finite in prop. 5
   %\State Choose $\theta_1=\argmin_{\theta'\in \Theta} D(P\|Q_{\theta'} \rightsquigarrow Q^*)$
   \State Let $Q_1=Q_{\theta_1}$, where $\theta_1=\displaystyle\argmin_{\theta'\in \Theta} D(P\|Q_{\theta'} \rightsquigarrow Q^*)$
   \For{$k=2,3,\dots $}
      \State Choose  $\alpha_k=\frac{2}{k+1}$ and $\theta_k= \argmin_{\theta'\in \Theta} D(P\|(1-\alpha_k) Q_{k-1}+\alpha_k Q_{\theta'}\rightsquigarrow  Q^*)$
      \State Let $Q_k=(1-\alpha_k) Q_{\theta_{k-1}}+\alpha_k Q_{\theta_k}$
   \EndFor
\end{algorithmic}
\end{algorithm}

\begin{proposition}\label{prop:algo_conv}
    Suppose that $\inf_{Q\in\mathcal{C}} D(P\|Q\rightsquigarrow \mathcal{C})<\infty$, let $(Q_k)_{k\in \mathbb{N}}$ be the output of Algorithm~\ref{alg:ripr}, and 
 let $Q$ be any measure in $\mathcal{C}$, so that  $q =\sum_{\theta \in \Theta'} q_\theta \cdot w_Q(\theta)$ for some probability mass function $w_Q$ on a countable $\Theta' \subset\Theta$.
    If $D(P\|Q'\rightsquigarrow Q'')$ is finite for all $Q',Q'' \in \mathcal{C}$,
    then it holds that
\[ D(P\|Q_k \rightsquigarrow Q) \leq \frac{b_Q^{(k)} (P)}{k},\]
    where
    $b_Q^{(k)}(P)$ is given by 
\begin{align} 
      & \int_\Omega \left( 1 + \sup_{\theta^* \in \{\theta_i\}_{i=1}^k} \log \frac{\sup_{\theta\in \Theta } q_\theta }{ q_{\theta^*} } \right) \frac{\sum_{\theta \in \Theta'} q_\theta^2 \cdot w_Q(\theta)}{q^2}\, \mathrm{d}P \leq   
   \nonumber \\ \label{eq:brindabound}
&    
\sup_{Q \in \mathcal{C}} \int_\Omega \left( 1 + \sup_{\theta^*,\theta \in \Theta} \log \frac{
q_\theta }{ q_{\theta^*} } \right) \frac{
\sum_{\theta \in \Theta'} q_\theta^2 \cdot w_Q(\theta)}{q^2}\, \mathrm{d}P. 
    \end{align}
\end{proposition}
% Example of when D(P\|\mathcal{C})=\infty, but D(P\|Q\rightsquigarrow Q')<\infty for all Q,Q': take P to be a t-distribution with df > 1 and \mathcal{C} = conv{N(\mu,1): \mu \in \mathbb{R}}. 
It follows that if 
$b_Q^{(k)}$ is uniformly bounded over all
$Q\in \mathcal{C}$,
in particular if (\ref{eq:brindabound}) is finite, 
% Tyron review universal removed
then $D(P\|Q_k\rightsquigarrow \mathcal{C})$ converges to zero, i.e. $Q_k$ converges to the RIPr of $P$ on $\mathcal{C}$, at rate $\nicefrac{1}{k}$. 
%
% Tyron review added
The former holds under the strong, but often imposed assumption that the likelihood ratios in $\mathcal{C}$ are uniformly bounded; for example when $\mathcal{C}$ is given by the $\sigma$-convex hull of Gaussian densities restricted to a cube~\cite[Example 1]{li1999Estimation}.
However, \eqref{eq:brindabound} might also be finite under weaker assumptions.
For example, consider the set of Gaussian mixtures as in Example~\ref{ex:gauss_cauchy}, that is, $\mathcal{C}=\{w\cdot \mathcal{N}(-1,1)+(1-w)\cdot \mathcal{N}(1,1): w\in [0,1]\}$.
It can be seen that $b_Q(P)<\infty$ for all $Q\in \mathcal{C}$ whenever $P$ has a finite first moment.
Moreover, if the latter holds, then $b_Q(P)$ is uniformly bounded over all $Q\in \mathcal{C}'$ where $\mathcal{C}'=\{w\cdot \mathcal{N}(-1,1)+(1-w)\cdot \mathcal{N}(1,1):w\in [c,1-c]\}$ for some $c\in (0,\nicefrac12)$.

While Proposition~\ref{prop:algo_conv} is a satisfying theoretical result, we must concede that Algorithm~\ref{alg:ripr} might not be the fastest to implement in practice.
This arises from the fact that the objective $D(P\| (1-\alpha_k)Q_{{k-1}} + \alpha_k Q_{\theta'} \rightsquigarrow Q^*)$ need not be convex in $\theta'$. 
One might therefore have to resort to an exhaustive search over a discretization of the parameter space.
On top of that, there is no guarantee that the information gain is easily computable.
As an alternative for the case that $\Theta$ is finite and $D(P\|\mathcal{C})<\infty$, one might use the iterative algorithm proposed by Csisz\'ar and Tusn\'ady~\cite[Theorem~5]{csiszar1984information}.
A big advantage of the latter is that their recursive update step has an explicit formula, which makes each iteration considerably faster. 
The downside is that, while convergence in terms of KL is guaranteed, it is unclear at what rate this happens in general.
Furthermore, proving convergence of their algorithm in the setting where $D(P\|\mathcal{C})=\infty$ seems far from a straightforward exercise.  

\subsection{Discussion}
The results in this section might be regarded as a generalization of large parts of Chapters 3 and 4 in 
Li's Ph.D. thesis \cite{li1999Estimation} and in fact the tools in this section were initially developed to clear up some ambiguity around 
the proof of 
Theorem~\ref{thm:li}, Part \ref{item:ripr_conv} as provided by Li.
That is, in~\cite{li1999Estimation} it is stated that for all sequences $(Q_n)_{n\in \mathbb{N}}$ in $\mathcal{C}$ such that $\lim_{n\to \infty} D(P\|Q_n)=D(P\|\mathcal{C})$ it holds that $\ln q_n\to \ln \hat q$ in $L_1(P)$.
However, the proof thereof refers to  \cite[Lemma~4.3]{li1999Estimation}, which only shows existence of one such a sequence. 
Then \cite[Lemma~4.4]{li1999Estimation} also shows that if $\hat Q$ is such that  $\log q_n \to \log \hat q$ in $L_1(P)$ for some sequence $(Q_n)_{n\in \mathbb{N}}$ that achieves 
$\lim_{n\rightarrow \infty} D(P\|Q_n)= D(P\|\mathcal{C})$,
then it must hold that $D(P\|\hat Q)=D(P\|\mathcal{C})$. 
However, it is a priori not clear whether every sequence $(Q_n)_{n\in \mathbb{N}}$ that achieves $\lim_{n\rightarrow \infty} D(P\|Q_n)= D(P\|\mathcal{C})$ has such a limit. 
Moreover, it is never shown that, if it exists, this limit must be the same for every such sequence.
Note that it is not at all our intention here to criticize Li's fundamental and ground-breaking work. Li's is one of those rare theses that have had a major impact outside of their own research area: being a thesis on information-theory, it served as the central tool and inspiration for papers on fast convergence rates in machine learning theory \cite{erven2015fast,GrunwaldM20}, and also for  \cite{grunwald2024}, which  led to a breakthrough in ($\E$-based) hypothesis testing. Our aim is merely to indicate that Theorem~\ref{thm:InfoGain} ties up some loose ends in Li's original, pioneering results. 
\section{Optimal \texttt{E}-statistics}\label{sec:ordering}
In this section, we assume that $P$ and all $Q\in \mathcal{C}$ are probability measures, and we are interested in the hypothesis test with $P$ as alternative and $\mathcal{C}$ as null.
% Tyron review universal removed, if it exists -> whenever it exists
To this end, Theorem~\ref{thm:InfoGain} shows that --- whenever it exists --- the likelihood ratio of $P$ and its RIPr is an \E-statistic. 
A natural question is whether the optimality of the RIPr in terms of describing data distributed according to $P$ carries over to some sort of optimality of the \E-statistic, as is true for the GRO criterion in the case that $D(P\|\mathcal{C})<\infty$. 
It turns out that this is true in terms of an intuitive extension of the GRO criterion.
Completely analogously
to the coding story, we simply have to change from absolute to pairwise comparisons.

\begin{definition}\label{def:ordering}
For \E-statistics $E,E'\in \mathcal{E}_\mathcal{C}$, we say that $E$ is \emph{stronger} than $E'$ if the following integral is well-defined and non-negative, possibly infinite:
\begin{equation}\label{eq:ordering}
\int_{\Omega} \ln\left(\frac{E}{E'}\right) \, \mathrm{d}P,
\end{equation}
where we adhere to the conventions $\ln(\nicefrac{0}{c})=-\infty$ and $\ln(\nicefrac{c}{0})=\infty$ for all $c\in\mathbb{R}_{>0}$. 
Furthermore, an \E-statistic $E^*\in \mathcal{E}_\mathcal{C}$ is a \emph{strongest} \E-statistic if it is stronger than any other \E-statistic $E\in \mathcal{E}_\mathcal{C}$.  
\end{definition} 
The notion of optimality in Definition~\ref{def:ordering} comes down to 
the simple idea that if one \E-statistic $E$ is stronger than another 
\E-statistic $E'$, then  {\em repeatedly\/} testing based on $E$ 
eventually becomes more powerful than repeatedly testing based on $E'$ 
in the sense that 
there is a higher probability of rejecting a false null-hypothesis. 
Let us explain in more detail what we mean by this.
Suppose that we conduct the same experiment $N$ times independently to 
test the veracity of the hypothesis $\mathcal{C}$, resulting in outcomes 
$\omega_1,\dots,\omega_N$.
For any given \E-statistic $E \in \mathcal{E}_{\mathcal{C}}$, we have that 
$\prod_{i=1}^N E(\omega_i)$
is still an e-statistic, not just for fixed $N$ but even if $N$ is a 
random (i.e. data-dependent) stopping time. So, as indicated before, it 
can be used to test $\mathcal{C}$ with Type-I error guarantees. Yet, for 
two \E-statistics $E,E'\in \mathcal{E}_\mathcal{C}$, the law of large 
numbers states that if $P$ is true, it will almost surely hold that 
\[
\frac{\prod_{i=1}^n E(\omega_i)}{\prod_{i=1}^n E'(\omega_i)} = \exp\left(n\int_\Omega \ln \left(\frac{E}{E'}\right) \, \mathrm{d}P+o(n)\right).
\]
It follows that if the integral 
$\int_\Omega \ln \left(\frac{E}{E'}\right) \, \mathrm{d}P$ 
is positive then with high probability $E$ will, for large enough $n$, 
give more evidence against $\mathcal{C}$ than $E'$ if the alternative is 
true, i.e.\ a test based on $E$ will asymptotically have more power than a 
test based on $E'$.

Since we assume throughout that there exists a $Q^*\in \mathcal{C}$ such 
that $P\ll Q^*$, it follows that for any \E-statistic $E$ we must have 
$P(E=\infty)=0$, which simplifies any subsequent analyses greatly. 
\begin{proposition}\label{prop:unique}
    Assume that $\mathcal{C}$ is a set of probability measures and that 
    $P$ is a probability measure. If there is an $E'\in \mathcal{E}_\mathcal{C}$ such that $\sup_{E\in \mathcal{E}_\mathcal{C}} \int_{\Omega} \ln\left(\nicefrac{E}{E'}\right) \, \mathrm{d}P <\infty$,
    then a strongest \E-statistics exists. 
    Furthermore, if $E_1$ and $E_2$ are both 
    strongest \E-statistics then $E_1=E_2$ holds $P$-a.s.
\end{proposition}

% Tyron review more elaborate
The strongest \E-statistic in Definition~\ref{def:ordering} can be seen as 
a generalization of the GRO \E-statistic, because if 
$\int_\Omega \ln E\, \mathrm{d}P$ and $\int_\Omega \ln E'\, \mathrm{d}P$ 
are both finite, \eqref{eq:ordering} can be written as the difference 
between the two logarithms, i.e. 
$\int_\Omega \ln\left(\nicefrac{E}{E'}\right)\,\mathrm{d}P=\int_\Omega \ln E\, \mathrm{d}P-\int_\Omega \ln E'\, \mathrm{d}P$. 
In this case, finding the strongest e-statistic therefore corresponds to 
maximizing $\int_\Omega \ln E\, \mathrm{d}P$ over all \E-statistics, thus 
recovering the original GRO criterion. 
%
% Tyron review universal removed
As an extension of that case, we prove that whenever the RIPr exists, it 
always leads to the strongest \E-statistic. 

\begin{theorem}\label{thm:LR_optim}
    Suppose that both $P$ and all $Q\in \mathcal{C}$ are probability 
    measures and that 
    $\inf_{Q\in\mathcal{C}}D(P\|Q\rightsquigarrow \mathcal{C})<\infty$.
    If $\hat Q$ denotes the RIPr of $P$ on $\mathcal{C}$, then 
    $\hat E=\nicefrac{\mathrm{d}P}{\mathrm{d}\hat Q}$ is the strongest 
    \E-statistic. 
\end{theorem}

% Tyron review universal removed
The likelihood ratio between $P$ and its RIPr is in fact the only 
\E-statistic in the form of a likelihood ratio with $P$ in the numerator, 
as the following proposition shows. 
Though the statement is more general, the proof is completely analogous to 
part of the proof of Lemma 4.1 in~\cite{li1999Estimation}.

\begin{proposition}\label{prop:estat_implies_ripr}
    % Tyron review universal removed
    Suppose that $\mathcal{C}$ is a set of probability measures and that 
    $P$ is a probability measure. 
    If there exists a measure $Q^*\in \mathcal{C}$ such that 
    $\nicefrac{\mathrm{d}P}{\mathrm{d}Q^*} \in \mathcal{E}_\mathcal{C}$, 
    then $D(P\|Q^*\rightsquigarrow \mathcal{C})=0$, i.e. $Q^*$ is the RIPr 
    of $P$ on $\mathcal{C}$. 
\end{proposition}

We now return to Example~\ref{ex:gauss_cauchy}, where the GRO criterion is 
not able to distinguish between \E-variables, but we are able to do so 
with Definition~\ref{def:ordering} and Theorem~\ref{thm:LR_optim}. 

\begin{excont}[continued]\label{ex:simple}
In the case that $P$ is the standard Cauchy and $\mathcal{C}=\{Q\}$, where 
$Q$ is the standard Gaussian, it is straightforward to see that the 
likelihood ratio between $P$ and $Q$ is an \E-statistic, i.e.
\[\int_\Omega \frac{\mathrm{d}{P}}{\mathrm{d}Q}\, \mathrm{d}Q= \int_\Omega\,\mathrm{d}P =1.\]
However, for the growth rate it holds that
\[\int_\Omega \ln\left(\frac{\mathrm{d}{P}}{\mathrm{d}Q}\right)\, \mathrm{d}P= D(P\|Q)=\infty.\]
The same argument can be used to show that for any $0< c\leq 1$, we have 
an \E-statistic given by $c\,{\mathrm{d}P}/{\mathrm{d}Q}$, which still has 
infinite growth rate.
The GRO criterion in Definition~\ref{def:GHKgro} is not able to tell which of these \E-statistics is preferable.
% Tyron review universal removed
However, since $Q$ is the RIPr of $P$ on $\mathcal{C}$, it follows from Theorem~\ref{thm:LR_optim} that $\nicefrac{\mathrm{d}P}{\mathrm{d}Q}$ is the strongest \E-statistic, and in particular it is stronger than $c\, {\mathrm{d}P}/{\mathrm{d}Q}$ for all $0<c<1$. 
 \end{excont}

\subsection{Convexity}\label{sec:convexity}
In the discussion above, the null hypothesis $\mathcal{C}$ is assumed to be convex, which does not hold for many of the null hypotheses commonly employed in statistics, such as the set of all Gaussian distributions with varying mean and/or variance.
However, it follows from the Fubini-Tonelli theorem that the set of \E-statistics on $\mathcal{C}$ equals the set of \E-statistics on the convex hull of $\mathcal{C}$.
The same is true if the convex hull is replaced by the $\sigma$-convex hull where countable mixtures are allowed or by the Choquet-convex hull where arbitrary mixtures are allowed (see Appendix~\ref{ap:convexity} for precise definitions). 
Therefore, if there exists a strongest \E-statistic for testing the alternative $P$ against any of these notions of the convex hull, then that is also the strongest \E-statistic for testing $P$ against the original null hypothesis, regardless of whether that was convex.
It follows from Theorem~\ref{thm:LR_optim} that, to find the strongest \E-statistic, it suffices to find the RIPr of $P$ on any of the notions of the convex hull. In particular, if RIPrs exists on more than one of these, they must coincide; on the other hand, none of the three RIPrs may exist, and our results also do not rule out the possibility that the  RIPr exists on just one or two of the three convex hulls. To witness, in Appendix~\ref{ap:convexity} we give an example (Example~\ref{ex:ripr_coincides}) in which the
RIPr of $P$ on the $\sigma$-convex hull exists, while the RIPr of $P$  on the convex hull does not. At the same time, there are constraints here: Theorem~\ref{thm:ripr_convexity} in Appendix~\ref{ap:convexity} implies that if the RIPr on the convex hull of $\mathcal{C}$ exists, then the RIPr on the $\sigma$-convex hull of $\mathcal{C}$ also exists (and then they must be equal). 

Things become much more clear-cut if the RIPr $\hat{Q}$ of $P$ on a certain notion of the convex hull exists {\em and is an element of that set}. In that case, $\hat Q$ is also the RIPr of $P$ on any stronger notion of the convex hull.
Indeed, the different levels of convex hulls are nested, and their corresponding sets of \E-statistics coincide, so this follows directly from Proposition~\ref{prop:estat_implies_ripr}:
\begin{corollary}\label{cor:ripr_convexity}
    Let $\mathcal{C}$ denote a set of probability measures (not necessarily convex) and let $P$ denote a probability measure. 
    If the RIPr of $P$ on the convex hull of $\mathcal{C}$ exists and is given by $\hat{Q}\in \mathrm{conv}(\mathcal{C})$, then  $\hat{Q}$ is also, (a) the RIPr of $P$ on the $\sigma$-convex and, (b), on the Choquet-convex hull of $\mathcal{C}$.
    Similarly, if $\hat Q\in \sigma\text{-}\mathrm{conv}(\mathcal{C})$ is the RIPr of $P$ on $\sigma\text{-}\mathrm{conv}(\mathcal{C})$, then (c) $\hat Q$ is also the RIPr of $P$ on the Choquet-convex hull of $\mathcal{C}$.
\end{corollary}
Further details regarding convexity are presented in Appendix \ref{ap:convexity}. 
In particular, Theorem~\ref{thm:ripr_convexity} in the latter gives an analogous result to 
Corollary~\ref{cor:ripr_convexity}, Part (a), for the case that $P$ and $\mathcal{C}$ are not restricted to 
be probability measures, and the RIPr is not assumed to be attained in the set.

\subsection{Approximation}\label{sec:approximation}
% Tyron review universal removed
In Section~\ref{sec:algo}, we discussed an algorithm that provides an approximation of the RIPr for scenarios where it is not possible to explicitly compute the latter. 
However, the convergence guarantee given by Proposition~\ref{prop:algo_conv} is in terms of the information gain.
That is, if $Q_k$ is the approximation of the projection after $k$ iterations, then under suitable conditions it holds that $D(P\|Q_k\rightsquigarrow \mathcal{C})\to 0$.
This is not enough if we want to use such an approximation for hypothesis testing: we need that $\nicefrac{p}{q_k}$ gets closer and closer to being an \E-statistic. 
The following theorem gives a condition under which this is true.
For $p\in(0,\infty]$ we use $\|f\|_p$ to denote the $\mathcal{L}^p(\Omega, P)$ norm of a function $f\in \mathcal{M}(\Omega, \mathbb{R}_{>0})$, i.e. $\left(\int_\Omega (f)^p \, \mathrm{d}P\right)^{\nicefrac1p}$.

\begin{theorem}\label{thm:approx}
Assume that $\inf_{Q\in\mathcal{C}} D(P\|Q\rightsquigarrow \mathcal{C})<\infty$, fix $Q,Q' \in \mathcal{C}$, set $\delta:=D(P \| Q\rightsquigarrow\mathcal{C})$ and suppose that there exists $\beta\in (0,\infty]$ such that  $\|\nicefrac{q'}{q}\|_{1+\beta}<\infty$.
If $\beta\leq 1$ or $D(P\|Q'\rightsquigarrow \mathcal{C})\leq K\delta$, then it holds that 
\begin{equation}\label{eq:rateofconv}
    \int_\Omega \frac{p}{q} \, \mathrm{d}Q' = \int_\Omega \frac{q'}{q}\, \mathrm{d}P= 1+O\left( C_{\beta} \cdot \delta^{\frac{\beta}{1+\beta}}\right) \text{ as } \delta \to 0,  % 1 + c_\beta \delta^{\nicefrac{\beta}{(1+\beta)}}+ 2\max \left\{ \left\|\frac{q'}{q}\right\|_\beta \left(2 \delta\right)^{\nicefrac{\beta}{1+\beta}} \left( \frac{\beta}{1-\beta}\right)^{\frac{1-\beta}{1+\beta}},2\delta \right\},  where c_\beta:= 2^{\frac{\beta}{1+\beta}}\|\frac{q'}{q}\|_\beta^\beta \left(\frac{\beta}{1-\beta} \right)^{-2\beta/(1+\beta)}, c_1 := 0 and c_1 \frac{\beta}{1-\beta} is understood as 1. 
\end{equation}
where $C_{\beta} = \|\nicefrac{q'}{q}\|_{1+\beta}$ if $\beta \leq 1$ and $C_{\beta}= K^{\frac{\beta-1}{2(1+ \beta)}} \|\nicefrac{q'}{q}\|_{1+\beta}$ otherwise. 

\end{theorem}
Explicit values for the constants in~\eqref{eq:rateofconv} can be found in the proof in the appendix.
In particular, Theorem~\ref{thm:approx} implies the following: if there are $C,\delta_0 > 0$ such that $\|\nicefrac{q'}{q}\|_2 \leq C$ for all $Q' \in \mathcal{C}$ and all $Q\in \mathcal{C}$ with $D(P \|Q\rightsquigarrow\mathcal{C}) \leq \delta_0$,
then any sequence $Q_1, Q_2, \ldots$ with $D(P \| Q_k\rightsquigarrow\mathcal{C}) \rightarrow 0$ will have $\sup_{q'\in \mathcal{C}}\int_\Omega \nicefrac{q'}{q_k} \, \mathrm{d}P= 1 + O(\delta_k^{\nicefrac{1}{2}})$, where $\delta_k = D(P \| Q_k\rightsquigarrow\mathcal{C})$.
This gives an easy to check condition for the convergence of $\nicefrac{p}{q_k}$ to an \E-statistic. 
\commentout{\peterg{We note that in the limit for $\beta \downarrow 0$, the dependence of $c_2$ on $K$ in  (\ref{eq:rateofconvB}) disappears and the rate obtained by (\ref{eq:rateofconvB}) becomes equal to the rate (\ref{eq:rateofconv}) but with better constants. However, the proof of (\ref{eq:rateofconvB}) does not seem to extend to the limiting case $\beta=0$. 
We do conjecture though that the constant factors in~\eqref{eq:rateofconv}  can be improved to (at least) those of (\ref{eq:rateofconvB}), and the latter constants may be improvable as well} --- our aim here was to give a result with a simple proof that provides a convergence rate. }
This square-root rate of convergence cannot be improved in general without an extra assumption, even if all likelihood ratios are bounded, i.e. $\|\nicefrac{q'}{q}\|_\infty < \infty$.
This can be seen by taking $P$ and $Q$ to be Bernoulli distributions with parameter $\nicefrac{1}{2}$ and $\nicefrac{1}{2} + \epsilon$ respectively, $\mathcal{C}$ the set of Bernoulli distributions with parameters in $[\nicefrac{1}{4},\nicefrac{3}{4}]$  and $Q'$ Bernoulli $\nicefrac{1}{4}$. 
Then $\delta = D(P \| Q\rightsquigarrow\mathcal{C}) = 2 \epsilon^2 (1+ o(1))$ yet $\int_\Omega \nicefrac{q'}{q}\, \mathrm{d}P =  1 + 4 \epsilon (1 + o(1))$.
But if likelihood ratios are bounded and we additionally consider $Q'$ in a `neighborhood' of $Q$ (i.e. $D(P \| Q'\rightsquigarrow\mathcal{C}) \leq K \delta$), 
then a linear rate is possible as shown in Theorem~\ref{thm:approx} by letting $\beta$ tend to infinity; the rate then interpolates between $\delta^{1/2}$ and $\delta$ depending on the largest $\beta$ for which the $(1+\beta)$-th moment exists.
Furthermore the following example shows that in general bounds on the integrated likelihood ratios are necessary for the convergence to hold at all. 

\begin{example}\label{ex:notApproxEvariable}
Let $\mathcal{Q}$ represent the family of geometric distributions on $\Omega = \mathbb{N}_0$ and let $\mathcal{C}=\mathrm{conv}(\mathcal{Q})$. 
The elements of $\mathcal{Q}$ are denoted by $Q_\theta$ with density $q_{\theta}(n) = \theta^n (1-\theta)$, where $\theta\in [0,1)$ denotes the probability of failure. 
For simplicity, assume that $P \in \mathcal{Q}$ so that the reverse information projection of $P$ on $\mathcal{C}$ is equal to $P$. 
Take for example $P= Q_{\nicefrac{1}{2}}$, then for any $\theta,\theta'\in [0,1)$
\begin{align}\label{eq:geometric}
\int_\Omega \frac{q_{\theta'}}{q_{\theta}} \, \mathrm{d}P
&= \sum_{n=0}^{\infty} \left(\frac{1}{2} \frac{\theta'}{\theta}\right)^n 
\frac{1}{2} \frac{1-\theta'}{1- \theta} 
% for twocolumn  \nonumber \\ &
 = \begin{cases}
\frac{1}{1- \frac{\theta'}{2 \theta}} \cdot \frac{1}{2} \frac{1-\theta'}{1- \theta}, & \text{\ if\ } \theta' < 2 \theta; \\
\infty, & \text{\ otherwise;} \end{cases}
    \end{align}
whereas
\begin{align*}
% for two column & 
D(P \| Q_{\theta}) 
& 
= \sum_{n=0}^{\infty} \left(\frac{1}{2}\right)^{n+1}
\left(- n \log (2 \theta)  - \log 2 (1-\theta)  \right)= 
% two column \\  &  
\log \frac{\nicefrac{1}{2}}{\theta} \cdot  \sum_{n=1}^{\infty} n  \left(\frac{1}{2}\right)^{n+1}   + \log \frac{\nicefrac{1}{2}}{1- \theta}  \cdot \sum_{n=0}^{\infty} \left(\frac{1}{2}\right)^{n+1} 
\\ & = 
\log \frac{\nicefrac{1}{2}}{1- \theta} + \log \frac{\nicefrac{1}{2}}{\theta},
\end{align*}
Now consider a sequence $\nicefrac{1}{3} < \theta_1 < \theta_2 < \theta_3 \ldots $ that converges to $\nicefrac{1}{2}$. 
Then by the above, 
\[D(P\|Q_{\theta_i}) \to 0= D(P\|\mathcal{C}).\] 
We also see that for all $i$ and all $\theta'\in [2\theta_i, 1)$, we have 
\[
\int_\Omega \frac{q_{\theta'}}{q_{\theta_i}} \mathrm{d}P=\infty,
\]
i.e.\ for all $i$ we have $\sup_{\theta'\in [0,1)}\int_\Omega \nicefrac{q_{\theta'}}{q_{\theta_i}} \mathrm{d}P=\infty$.
%This shows  that in general, a condition such as \eqref{eq:theapproxcondition} is necessary. 
\end{example}

\subsection{Related Work}
The results on the existence of optimal \E-statistics displayed in this section bear similarities with work concurrently done by Zhang et al.~\cite{zhang2023exact}.
In particular, they show that if $\mathcal{C}$ is a convex polytope, then there exists an \E-statistic in the form of a likelihood ratio between two unspecified measures. 
Since a convex polytope contains the uniform mixture of its vertices, which can be shown to have finite information gain, this also follows from our Proposition~\ref{prop:convhull}.
However, the techniques used to prove their results appear to be of a completely different nature than the ones used in this paper, as they rely mostly on classical results in convex geometry together with results on optimal transport (and with these techniques, they provide various other results incomparable to ours).

In the case of compact alternative they furthermore discuss a property which they refer to as nontrivial \E-power.
That is, if the alternative is a convex polytope $\mathcal{A}$, then at least one of their \E-statistics in the form of a likelihood ratio satisfies $\inf_{P\in \mathcal{A}} \int_\Omega \ln E \, \mathrm{d}P > 0$.
We now show that the existence of such an \E-statistic also follows from our results. In fact, if
$\mathcal{A}$ is any convex set (not just a polytope) such that $\inf_{P\in \mathcal{A}} D(P\|\mathcal{C})<\infty$, then (as \cite{zhang2023exact} point out) a similar result is already implied by Gr\"unwald et al.~\cite{grunwald2024} as long as the infimum is achieved on the left. 
Indeed, they show that the likelihood ratio of the distribution that achieves the infimum and its RIPr is an \E-statistic that has nontrivial \E-power.
This leaves the case that $\inf_{P\in \mathcal{A}} D(P\|\mathcal{C})=\infty$. Indeed, the current work implies that also in this case, an \E-statistic with nontrivial (in fact, infinite) \E-power exists, as long as $\mathcal{A}$ is a convex polytope.
That is, if we use $P^*$ to denote the uniform mixture of the vertices of $\mathcal{A}$, then for any vertex $P\in\mathcal{C}$, we have that
\[ \int_\Omega \ln\frac{\mathrm{d}P^*}{\mathrm{d}\hat Q^*}\, \mathrm{d}P \geq \int_\Omega \ln\frac{\frac1n\mathrm{d}P}{\mathrm{d}\hat Q^*}\, \mathrm{d}P \geq D(P\|\mathcal{C})-\ln(n),  \]
% Tyron review universal removed
where $\hat Q^*$ denotes the RIPr of $P^*$. 
It follows that $\inf_{P\in\mathcal{A}}\int_\Omega \ln\frac{\mathrm{d}P^*}{\mathrm{d}\hat Q^*}\, \mathrm{d}P=\infty$, so that the \E-statistic given by the likelihood ratio of $P^*$ to its RIPr has ``nontrivial \E-power''.  
However, more work is needed to determine whether such constructions are in any way optimal and whether the restriction that $\mathcal{A}$ is a convex polytope can be relaxed.

% Tyron review added
Second, after the first version of this paper was made available online, a follow-up paper appeared by Larsson et al.~\cite{larsson2024numeraire}.
They show that, under no conditions on $P$ and $\mathcal{C}$ whatsoever, there exists an \E-statistic $E^*$ that is the strongest \E-statistic in the sense of Definition~\ref{def:ordering}.
This \E-statistic, which they call `the numeraire', gives rise to a measure $Q^*$ such that $\nicefrac{\mathrm{d}Q^*}{\mathrm{d}P}=\nicefrac{1}{E^*}$.
Whenever the conditions of Theorem~\ref{thm:InfoGain} hold, $Q^*$ coincides with the reverse information projection of $P$ on $\mathcal{C}$, so that it provides (in their words) ``[...] a natural definition of the RIPr in the absence of any assumptions on $\mathcal{C}$ or $P$.'' We refer to their work~\cite{larsson2024numeraire} for all further details. 

\section{Summary and Future Work}\label{sec:summary}
We have shown that, under very mild conditions, there exists a measure that achieves the minimax description gain over a convex set of measures $\mathcal{C}$ relative to a measure $P$.
Whenever the information divergence between $P$ and $\mathcal{C}$ is finite, this measure coincides with the reverse information projection of $P$ on $\mathcal{C}$.
As such, it provides a natural extension of the reverse information projection to cases where the the minimax description gain is finite, while the information divergence is infinite.
% Tyron review universal removed
In the context of hypothesis testing, this extended notion of the RIPr can be used to define an \E-statistic for testing the simple alternative $P$ against the composite null $\mathcal{C}$. 
This \E-statistic is optimal in a sense that is a natural, but novel extension of the previously known GRO optimality criterion for \E-statistics.
We have shown an example where GRO is unable to differentiate between \E-statistics, while our novel criterion can, so that it is a strict extension. 
% Tyron review universal removed
Additionally, we discussed an algorithm that can be used to approximate the reverse information projection in scenarios where it is not explicitly computable and show under what circumstances this also leads to an approximation of the optimal \E-statistic. 

The results presented thus far suggest various avenues for further research of which we  discuss two. 
% Tyron review universal removed
First, Theorem \ref{thm:InfoGain} is formulated for general measures so one may ask for an interpretation of the RIPr in the case that $P$ and $\mathcal{C}$ are not probability measures. 
If $\Omega$ is finite and $\lambda$ is a measure on $\Omega$, 
then we may define a probability measure $Po(\lambda)$ as the product measure
$Po(\lambda)=\bigotimes_{\omega\in\Omega}Po\left(\lambda(\omega)\right)$,
where $Po\left(\lambda(\omega)\right)$ denotes the Poisson distribution with mean $\lambda(\omega)$. 
With this definition we get
\[D(P\|Q\rightsquigarrow Q')=D(Po(P)\|Po(Q)\rightsquigarrow Po(Q')).\]
% Tyron review universal removed
Furthermore, it can be shown that if the RIPr $\hat Q$ of $P$ on $\mathcal{C}$ exists and is an element of $\mathcal{C}$, then $Po(\hat Q)$ is also the RIPr of $Po(P)$ on the convex hull of $\mathcal{C}':=\{Po(Q)|Q\in \mathcal{C}\}$. 
Consequently, $\nicefrac{Po(P)}{Po(\hat Q)}$ can be thought of as an \E-statistic for $\mathcal{C}'$.
More work is needed to determine whether this interpretation has any applications and if it can be generalized to arbitrary $\Omega$.

Second, even if $D(P\|\mathcal{C})=\infty$, the R\'enyi divergence 
$D_{\alpha}\left(P\|Q\right)$ (see e.g. \cite{Erven2014}) may be a well-defined non-negative real number for $\alpha\in(0,1)$ and $Q\in\mathcal{C}$. 
These R\'enyi divergences are jointly convex in $P$ and $Q$ \cite{Erven2014} and for 
each $0<\alpha<1$ one may define a reverse R\'enyi projection $\hat{Q}_{\alpha}$ of $P$ on $\mathcal{C}$ \cite{Kumar2016}. 
If it exists, it can be shown that this distribution will satisfy
\[
\int_{\Omega}\left(\frac{\mathrm{d}P}{\mathrm{d}\hat{Q}_{\alpha}}\right)^{\alpha}\,\mathrm{d}Q \leq 1
\]
for all $Q\in\mathcal{C}$, i.e. $\left(\nicefrac{\mathrm{d}P}{\mathrm{d}\hat{Q}_{\alpha}}\right)^{\alpha}$ is an \E-statistic.
% Tyron review universal removed
We conjecture that the projections $\hat{Q}_{\alpha}$ will converge to the RIPr for $\alpha$ tending to 1, which might lead to further applications. 
%

%\section*{Acknowledgments}

\newpage

%%%%%%
%% References:
%% We recommend the usage of BibTeX:
%%
\bibliographystyle{IEEEtran}
\bibliography{proper_bib,master,peter,database1,phd_basic}
%\bibliography{definitions,bibliofile}
%%
%% where we here have assume the existence of the files
%% definitions.bib and bibliofile.bib.
%% BibTeX documentation can be obtained at:
%% http://www.ctan.org/tex-archive/biblio/bibtex/contrib/doc/
%%%%%%
%% Or you use manual references (pay attention to consistency and the
%% formatting style!):

%\begin{thebibliography}{9}

%\bibitem{IEEE:AuthorToolbox}
%IEEE, \emph{Author Center.} [Online.] Available:
%  \url{https://ieeeauthorcenter.ieee.org/}
%
%\end{thebibliography}

\newpage
\,
%\newpage
%%%%%% 
%% Appendix:
%% If needed a single appendix is created by
%%
\appendices
%%
%% If several appendices are needed, then the command
%%
% \appendices
%%
%% in combination with further \section-commands can be used.
%%%%%%

%The appendix (or appendices) are optional. For reviewing purposes
%additional 5~pages (double-column) are allowed (resulting in a maximum
%grand total of 10~pages plus one page containing only
%references). These additional 5~pages must be removed in the final
%version of an accepted paper.

\section{Proofs}\label{app:proofs}
\subsection{Proofs for Section~\ref{sec:universal_ripr}}\label{app:li_proof}
Before giving the intended results, we note that we introduced $m_P$ as the averaged Bregman divergence associated with $\gamma(x)=x-1-\ln(x)$. 
For the proof, it will be useful to also define the Bregman divergence associated with $\gamma(x)=x-1-\ln(x)$ itself, which is the so-called Itakura-Saito divergence.
For $f, g\in \mathcal{M}\left(\Omega,\mathbb{R}_{>0}\right)$, it is given by
\[IS_P(f, g)=\int_\Omega \left(\frac{f}{g} -1 - \ln \frac{f}{g}\right) \, \mathrm{d}P.\]
By definition, it holds that 
\[
m_P^2(f, g)=\frac12 IS\left(f, \frac{f+g}{2}\right)+\frac12 IS\left(g, \frac{f+g}{2}\right).
\] 
Furthermore, for $Q\in \mathcal{C}$, we have $IS_P(q, p)=D(P\|Q)$.
We now state some auxiliary results before giving the proofs for Section~\ref{sec:universal_ripr}.

\begin{lemma}\label{lem:gamma_transform}
For $x,y\in \mathbb{R}_{>0}$, we have
\[\left|\ln(x)-\ln(y)\right|=g(m_\gamma^2(x,y)),\] 
where $g$ denotes the function
\[g(t)=2t+2\ln\left(1+\left(1-\exp\left(-2t\right)\right)^{\nicefrac{1}{2}}\right).\]
The function $g$ is concave and satisfies $g(t)\geq 2t$.
\end{lemma}
\begin{proof}
Let $m=\frac{x+y}{2}$. 
Our goal is to determine the function $g$ function such that
\[\left|\ln(x)-\ln(y)\right|=g(m^2_\gamma(x,y)).\]
We first rewrite the right-hand side
\begin{align*}
g(m^2_\gamma(x, y)) & =g\left(\ln\left(m\right)-\frac{1}{2}\ln\left(x\right)-\frac{1}{2}\ln\left(y\right)\right)\\
& =g\left(\frac{1}{2}\ln\left(\frac{m^{2}}{x\cdot y}\right)\right)\\
& =g\left(\frac{1}{2}\ln\left(\frac{\left(\frac{m}{y}\right)^{2}}{\frac{x}{y}}\right)\right)\\
& =g\left(\frac{1}{2}\ln\left(\frac{\left(\frac{1+\frac{x}{y}}{2}\right)^{2}}{\frac{x}{y}}\right)\right).
\end{align*}
Plugging this back in and replacing $\frac{x}{y}$ by $w$ leads to
\begin{equation*}
    \left|\ln\left(w\right)\right|  =g\left(\frac{1}{2}\ln\left(\frac{\left(\frac{1+w}{2}\right)^{2}}{w}\right)\right)
\end{equation*}
Then we solve the equation
\begin{equation*}
    \frac{1}{2}\ln\left(\frac{\left(\frac{1+w}{2}\right)^{2}}{w}\right)  =t,
\end{equation*}
which gives
\begin{equation*}
    w =2\exp\left(2t\right)-1+2\cdot\left(\exp\left(4t\right)-\exp\left(2t\right)\right)^{\nicefrac{1}{2}}
\end{equation*}
\begin{align*}
g\left(t\right) & =\ln\left(2\exp\left(2t\right)-1+2\cdot\left(\exp\left(4t\right)-\exp\left(2t\right)\right)^{\nicefrac{1}{2}}\right)\\
 & =2t+\ln\left(2-\exp\left(-2t\right)+2\cdot\left(1-\exp\left(-2t\right)\right)^{\nicefrac{1}{2}}\right)\\
 & =2t+2\ln\left(1+\left(1-\exp\left(-2t\right)\right)^{\nicefrac{1}{2}}\right).
\end{align*}

The derivatives of $g$ are
\begin{align*}
    g'(t)&= 2+2\frac{\left(1-\exp\left(-2t\right)\right)^{\nicefrac{-1}{2}}\exp\left(-2t\right)}{1+\left(1-\exp\left(-2t\right)\right)^{\nicefrac{1}{2}}}=\frac{2}{(1-\exp(-2t))^{1/2}}\\
    g''(t)&=\frac{-\exp\left(-\nicefrac{t}{2}\right)}{2^{\nicefrac{1}{2}}\left(\sinh{t}\right)^{\nicefrac{3}{2}}}.
\end{align*}
We see that $g''(t)<0$ and conclude that $g$ is concave. Finally, we have 
\begin{align*}
g\left(t\right) & =2t+2\ln\left(1+\left(1-\exp\left(-2t\right)\right)^{\nicefrac{1}{2}}\right)\geq 2t,
\end{align*}
because $1-\exp\left(-2t\right)\geq 0$.
\end{proof}

\begin{lemma}\label{lem:aux2}
Let $(f_n)_{n\in \mathbb{N}}$ be a sequence of elements of $\mathcal{M}(\Omega, \mathbb{R}_{>0})$, then
\[\limsup_{m,n\to \infty} m_P(f_m,f_n) = 0 \Leftrightarrow \limsup_{m,n\to \infty}  \int_\Omega \left|\ln \left(\frac{f_m}{f_n}\right) \right|\, \mathrm{d}P= 0.\]
\end{lemma}
\begin{proof}
%Tyron review added
By Lemma~\ref{lem:gamma_transform}, we have for $m,n\in \mathbb{N}$,
\begin{align*}
    m_P^2(f_n,f_m)&=\int_\Omega m_\gamma^2(f_n,f_m) \, \mathrm{d}P\\
    &\leq \frac12 \int_\Omega \left|\ln\left(\frac{f_m}{f_n}\right)\right|\, \mathrm{d}P,
\end{align*}
as well as
\begin{align*}
    \int_\Omega \left| \ln \left(\frac{f_n}{f_m}\right) \right|\, \mathrm{d}P
    &= \int_{\Omega} g(m^2_\gamma(f_n , f_m)) \,\mathrm{d}P\\
    &\leq g\left(\int_{\Omega} m^2_\gamma(f_n , f_m) \,\mathrm{d}P\right)\\
    &= g\left(m^2_P\left(f_n , f_m\right)\right).
\end{align*}
The result then follows by continuity of $g$. 
\end{proof}

\begin{lemma}\label{lem:3pointNy}
For $Q_1 ,Q_2\in \mathcal{C}$ such that $P\ll Q_i$ for $i\in \{1,2\}$, we have
\[
m_P^2(q_1 , q_2) \leq \frac{D(P\|Q_1\rightsquigarrow \mathcal{C})+D(P\|Q_2\rightsquigarrow \mathcal{C})}{2}.
\] 
\end{lemma}
\begin{proof}
    Let $\bar{Q}$ denote the midpoint between $Q_1$ and $Q_2$. Then we have
    \begin{align*}
    \frac{D(P\|Q_1\rightsquigarrow \mathcal{C})+D(P\|Q_2\rightsquigarrow \mathcal{C})}{2} %\\
    &= \frac{\sup_{Q\in\mathcal{C}} D(P\|Q_1\rightsquigarrow Q)+\sup_{Q\in\mathcal{C}} D(P\|Q_2\rightsquigarrow Q)}{2}\\
    &\geq \frac{D(P\|Q_1\rightsquigarrow \bar{Q})+D(P\|Q_2\rightsquigarrow \bar{Q})}{2}
    =m_P^2(q_1 , q_2).
    \end{align*}   
\end{proof}

\begin{proof}[Proof of Proposition~\ref{prop:complete}]
This follows as a direct corollary of Lemma~\ref{lem:aux2}.
\end{proof}

We now deviate slightly from the order of the results in Section~\ref{sec:universal_ripr} and first state the proof of Proposition~\ref{prop:equvalence}, so that we can use its results in the proof of Theorem~\ref{thm:InfoGain}.

\begin{proof}[Proof of Proposition~\ref{prop:equvalence}]
    The implications $(3)\to(2)\to(1)$ are obvious, so we show here only the implication $(1)\to(3)$.
    Assume that $P'$ is a measure such that $-\infty <D(P\|P'\rightsquigarrow\mathcal{C})<\infty$. Then there exists a sequence of measures $Q_n\in \mathcal{C}$ such that 
    \[
    D(P\|P'\rightsquigarrow Q_n)\to D(P\|P'\rightsquigarrow\mathcal{C})
    \] 
    for $n\to\infty$. Without loss of generality we may assume that $-\infty <D(P\|P'\rightsquigarrow Q_n)<\infty$ for all $n$. The result follows because 
    \[
    D(P\|P'\rightsquigarrow\mathcal{C})=D(P\|P'\rightsquigarrow Q_n)+D(P\|Q_n\rightsquigarrow\mathcal{C})
    \]
    and all involved quantities are finite.
\end{proof}

\begin{proof}[Proof of Theorem~\ref{thm:InfoGain}~(\ref{item:conv_mpny})]
Let $(Q_n)_{n\in \mathcal{C}}$ denote a sequence in $\mathcal{C}$ such that 
\[
\lim_{n\to \infty} D(P\|Q_n\rightsquigarrow \mathcal{C})=\inf_{Q\in \mathcal{C}} D(P\|Q\rightsquigarrow \mathcal{C})=0,
\] 
where the last equality follows from Proposition~\ref{prop:equvalence}.
Without loss of generality, we may assume that $D(P\|Q_n\rightsquigarrow \mathcal{C})<\infty$ for all $n$, so that $P\ll Q_n$ for all $n$. 
It then follows from Lemma~\ref{lem:3pointNy} that for $m,n\in \mathbb{N}$ we have
%two column \begin{multline}
$$
m_P^2(q_m, q_n)\leq
%\\
 \frac{D(P\|Q_m\rightsquigarrow \mathcal{C})+D(P\|Q_n\rightsquigarrow \mathcal{C})}{2}.
%\end{multline}
$$
It follows that $(q_n)_{n\in \mathbb{N}}$ is a Cauchy sequence with respect to $m_P$, so that $(q_n)_{n\in \mathbb{N}}$ converges to some function $\hat q$ in $m_P$.
The latter follows from the completeness of $(\mathcal{M}\left(\Omega,(0,\infty)\right), m_P)$, i.e.\ Proposition~\ref{prop:complete}.

Furthermore, suppose that $(Q'_n)_{n\in \mathcal{C}}$ is another sequence in $\mathcal{C}$ such that 
\[\lim_{n\to \infty} D(P\|Q'_n\rightsquigarrow \mathcal{C})=0.\]
Then, by the same reasoning as before, $Q_1 ,Q'_1 ,Q_2 ,Q'_2 ,Q_3 ,Q'_3 ,\dots$ is also a Cauchy sequence that converges and since a Cauchy sequence can only converge to a single element this implies the desired uniqueness.
\end{proof}

\begin{proof}[Proof of Theorem~\ref{thm:InfoGain} (\ref{item:conv_lrny})]
The equality 
\[
\int_\Omega \ln\frac{p'}{\hat q}\,\mathrm{d}P=\lim_{n\to \infty} \int_\Omega  \ln \frac{p'}{q_n} \,\mathrm{d}P
\]
follows from Theorem~\ref{thm:InfoGain}~(\ref{item:conv_mpny}) together with the fact that convergence of $q_n$ in $m_P$ implies convergence of the logarithms in $L_1(P)$.
\end{proof}

\begin{proof}[Proof of Theorem~\ref{thm:InfoGain} (\ref{item:conv_evalny})]
Let $(Q_n)_{n\in \mathcal{C}}$ denote a sequence in $\mathcal{C}$ such that 
\[
\lim_{n\to \infty}D(P\|Q_n\rightsquigarrow \mathcal{C})=0.
\]
 Without loss of generality, we may assume that $D(P\|Q_n\rightsquigarrow \mathcal{C})<\infty$ for all $n$ and that $q_n$ converges to $\hat q$  $P$-almost surely. 
The latter is valid, because convergence in $m_P$ implies convergence of the logarithms in $L_1(P)$ by Lemma~\ref{lem:aux2}, which gives the existence of an almost surely converging sub-sequence.

Let $\tilde{Q}=(1-t)Q_1 +tQ$  for fixed $Q\in \mathcal{C}$ and fixed $0<t<1$. Let $Q_{n, s}$ denote the convex combination $Q_{n,s}=(1-s_n)Q_n+s_n \tilde{Q}$ and $s_n \in [0,1]$. 
By Theorem~\ref{thm:InfoGain}~(\ref{item:conv_mpny}), we know that there exists some $\hat Q$ such that $q_n \to \hat q$ in $m_P$. 

Since $Q_{n,s}\in\mathcal{C}$ by convexity, we have that $D(P\|Q_n\rightsquigarrow Q_{n,s})\leq D(P\|Q_n\rightsquigarrow \mathcal{C})$.
We also have
\begin{align*}
    D(P\|Q_n\rightsquigarrow Q_{n, s})&=
     s_nD(P\|Q_n\rightsquigarrow \tilde{Q})+s_nIS_P\left(\tilde{q}, q_{n, s}\right)+(1-s_n)IS_P\left(q_n, q_{n, s}\right)\\
    &\geq s_n D(P\|Q_n\rightsquigarrow \tilde{Q})+s_nIS_P(\tilde{q}, q_{n, s}).
\end{align*}
Hence
\[
s_n D(P\|Q_n\rightsquigarrow \tilde{Q})+s_nIS_P(\tilde{q}, q_{n, s}) 
    \leq D(P\|Q_n\rightsquigarrow \mathcal{C}).
\]
Division by $s_n$ gives
\[
D(P\|Q_n\rightsquigarrow \tilde{Q}) +IS_P(\tilde{q}, q_{n, s})\leq \frac{D(P\|Q_n\rightsquigarrow \mathcal{C})}{s_n}.
\]
Choosing $s_n=D(P\|Q_n\rightsquigarrow \mathcal{C})^{\nicefrac{1}{2}}$, this gives
\begin{align*}
    D(P\|Q_n\rightsquigarrow \tilde{Q})+IS_P(\tilde{q}, q_{n, s}) &\leq s_n^{\nicefrac{1}{2}}.
\end{align*}
Then we get
\begin{align*}
    IS_P(\tilde{q}, q_{n, s}) &\leq D(P\|\tilde{Q}\rightsquigarrow Q_n)+s_n^{\nicefrac{1}{2}}.\\
    \int_\Omega \left(\frac{\tilde{q}}{q_{n, s}}+\ln \frac{q_{n, s}}{q_n}\right)\,\mathrm{d}P &\leq P(\Omega)+\tilde{Q}(\Omega)-Q_n(\Omega)+s_n^{\nicefrac{1}{2}}.
\end{align*}
Writing $q_n$ as $\frac{q_{n,s}-s_n \tilde{q}}{1-s_n}$, we see
\begin{align*}
    \ln \frac{q_{n, s}}{q_n}&= \ln \frac{q_{n, s}}{\frac{q_{n, s}-s_n\tilde{q}}{1-s_n}}\\
    &=\ln(1-s_n)-\ln \frac{q_{n, s}-s_n\tilde{q}}{q_{n, s}}\\
    &=\ln(1-s_n)-\ln\left(1-s_n\frac{\tilde{q}}{q_{n, s}} \right)\\
    &\geq \ln(1-s_n)+s_n\frac{\tilde{q}}{q_{n, s}}.
\end{align*}
Hence 
\[
    \ln(1-s_n)+(1+s_n)\int_\Omega \frac{\tilde{q}}{q_{n, s}} \, \mathrm{d}P
    \leq P(\Omega)+\tilde{Q}(\Omega)-Q_n(\Omega)+s_n^{\nicefrac{1}{2}}.
\]
As $\lim_{n\to \infty} s_n=0$, taking the limit inferior as $n\to \infty$ on both sides gives
\[
\liminf_{n\to \infty} \int_\Omega \frac{\tilde{q}}{q_{n, s}} \, \mathrm{d}P\leq P(\Omega)+\tilde{Q}(\Omega)-\liminf_{n\to \infty}Q_n(\Omega).
\]
An application of Fatou's lemma gives 
\[
\int_\Omega \frac{\mathrm{d}P}{\mathrm{d}\hat Q}  \, \mathrm{d}\tilde{Q}
\leq P(\Omega)+\tilde{Q}(\Omega)-\liminf_{n\to \infty}Q_n(\Omega).
\]
Since $\tilde{Q}=(1-t)Q_1 +tQ$ we get the inequality
\begin{align*}
\int_\Omega \frac{\mathrm{d}P}{\mathrm{d}\hat Q}  \, \mathrm{d}\left(\left(1-t\right)Q_1+t Q\right) &\leq P(\Omega)+(1-t)Q_1 (\Omega) +tQ(\Omega)-\liminf_{n\to \infty}Q_n(\Omega),\\
\left(1-t\right)\int_\Omega \frac{\mathrm{d}P}{\mathrm{d}\hat Q}  \, \mathrm{d}Q_1
+t\int_\Omega \frac{\mathrm{d}P}{\mathrm{d}\hat Q}  \, \mathrm{d}Q 
&\leq P(\Omega)+(1-t)Q_1 (\Omega) +tQ(\Omega)-\liminf_{n\to \infty}Q_n(\Omega).
\end{align*}
Finally we let $t$ tend to one and obtain the desired result.
\end{proof}

\begin{proof}[Proof of Proposition~\ref{prop:convhull}]
    Let $Q\in \mathcal{C}$ arbitrarily. Then there exists a sequence $(w_i)_{i=1}^n$ in $[0,1]$ with $\sum_i w_i=1$ such that $Q=\sum_{i=1}^n w_i Q_i$.
    It follows that
    \begin{align*}
        D\left(P\| \frac{1}{n}\sum_i Q_i \rightsquigarrow Q\right) &= \int_\Omega \ln \frac{\sum_i w_i Q_i}{\frac{1}{n}\sum_i Q_i}\, \mathrm{d}P\\
        &\leq \int_\Omega \ln \frac{\max_i w_i \sum_i Q_i}{\frac1n \sum_i Q_i}\, \mathrm{d}P\\
        &= \ln(n)+\ln(\max_i w_i)\leq \ln(n).
    \end{align*}
    The proposition follows by taking the supremum over $Q$ on both sides. 
\end{proof}

\begin{proof}[Proof of Proposition~\ref{prop:minimax}]
    Since $Q^*$ is the normalized maximum likelihood distribution we have $\sup_Q \sup_{\omega} \ln{\frac{\mathrm{d}Q}{\mathrm{d}Q^*}} <\infty.$ In particular 
    \begin{align*}
    \sup_{Q\in \mathcal{C}} D(P\|Q^*\rightsquigarrow Q)&=\sup_{Q\in\mathcal{C}}\int_\Omega\ln{\frac{\mathrm{d}Q}{\mathrm{d}Q^*}}\,\mathrm{d}P\\
    &\leq \sup_{Q\in\mathcal{C}}\sup_\omega \ln{\frac{\mathrm{d}Q}{\mathrm{d}Q^*}(\omega)}<\infty.
    \end{align*} 
\end{proof}

\begin{proof}[Proof of Proposition~\ref{prop:algo_conv}]
    We can write
    \[D(P\|Q_\theta \rightsquigarrow Q^*)=D(P\|Q_\theta \rightsquigarrow Q)+D(P\| Q\rightsquigarrow Q^*).\]
    By assumption all terms are finite so that minimising $D(P\|Q_\theta \rightsquigarrow Q^*)$ over $\theta$ must be equivalent to minimising $D(P\| Q_\theta \rightsquigarrow Q)$ over $\theta$. 
    The same argument holds for step $5$ in Algorithm~\ref{alg:ripr}. 
    The result then follows from \cite[Theorem 3.0.13]{brinda2018adaptive}. 
    While the algorithm described there works by choosing $\theta_k$ to minimise $\int_\Omega \log ((1-\alpha_k) q_{\theta_{k-1}}+\alpha_k q_\theta)\,\mathrm{d}P$, the proof relies on \cite[Lemma 5.9]{li1999Estimation}, which indeed uses minimisation of $D(P\|(1-\alpha_k) Q_{\theta_{k-1}}+\alpha_k Q_{\theta} \rightsquigarrow Q)$  as described here.
\end{proof}

%\subsection{Proof of Theorem~\ref{thm:aftagende}}

\begin{proof}[Proof of Theorem \ref{thm:aftagende}]
For any $a\in\mathbb{R}$
we have 
\begin{equation}
    f_{0}\left(i\right)+a\cdot f_{1}\left(i\right)=f_{0}\left(i\right)\cdot\left(1+a\cdot\frac{f_{1}\left(i\right)}{f_{0}\left(i\right)}\right).
\end{equation}
 Since $\frac{f_{1}\left(i\right)}{f_{0}\left(i\right)}\to0$ for
$i\to\infty$ we have that $f_{0}\left(i\right)+a\cdot f_{1}\left(i\right)\geq0$
for $i$ sufficiently large. Therefore, we can apply Fatou's lemma
to the function and obtain 
\begin{align*}
\sum f_{0}\left(i\right)\cdot q^{*}\left(i\right)+a\cdot\sum f_{1}\left(i\right)\cdot q^{*}\left(i\right)
&=\sum\left(f_{0}\left(i\right)+a\cdot f_{1}\left(i\right)\right)\cdot q^{*}\left(i\right)\\
  &=\sum\liminf_{n\to\infty}\left(f_{0}\left(i\right)+a\cdot f_{1}\left(i\right)\right)\cdot q_{n}\left(i\right)\\
  &\leq\liminf_{n\to\infty}\sum_{i}\left(f_{0}\left(i\right)+a\cdot f_{1}\left(i\right)\right)\cdot q_{n}\left(i\right)\\
  &=\liminf_{n\to\infty}\left(\sum_{i}f_{0}\left(i\right)\cdot q_{n}\left(i\right)+a\cdot\sum_{i}f_{1}\left(i\right)\cdot q_{n}\left(i\right)\right)\\
  &=\liminf_{n\to\infty}\left(\lambda_{0}+a\cdot\lambda_{1}\right)
  =\lambda_{0}+a\cdot\lambda_{1}.
\end{align*} 
Hence 
\begin{equation}
    a\cdot\left(\sum f_{1}\left(i\right)\cdot q^{*}\left(i\right)-\lambda_{1}\right)\leq\lambda_{0}-\sum f_{0}\left(i\right)\cdot q^{*}\left(i\right).
\end{equation}
 This inequality should hold for all $a\in\mathbb{R}$, which is only
possible if 
\begin{align*}
\sum f_{1}\left(i\right)\cdot q^{*}\left(i\right)-\lambda_{1} & =0.\\
\sum f_{1}\left(i\right)\cdot q^{*}\left(i\right) & =\lambda_{1}.
\end{align*}
\end{proof}

\subsection{Proofs for Section~\ref{sec:ordering}}

%\subsection{Proof of Theorem~\ref{thm:LR_optim}}
\begin{proof}[Proof of Proposition \ref{prop:unique}]
    Assume that $E_1 ,E_2 ,E_3 ,\dots$ is a sequence of \E-variables such that 
    \begin{equation*}
        \int_{\Omega} \ln\left(\frac{E_n}{E'}\right) \, \mathrm{d}P
        \to\sup_E \int_{\Omega} \ln\left(\frac{E}{E'}\right) \, \mathrm{d}P
    \end{equation*}
    for $n\to\infty$. Then $E_{n,m}=\left(E_m +E_n\right)/2$ are also \E-variables and by convexity
    \begin{equation*}
        \int_{\Omega} \ln\left(\frac{E_{m,n}}{E'}\right) \, \mathrm{d}P
        \to\sup_E \int_{\Omega} \ln\left(\frac{E}{E'}\right) \, \mathrm{d}P\, ,
    \end{equation*}
    which implies that $m_{\gamma}^2\left(E_m ,E_n\right)\to 0$ for 
    $m,n\to\infty$. By completeness $E_n$ converges to some \E-variable 
    $E_{\infty}$. Using Lemma \ref{lem:aux2} we see that
    $m_{\gamma}\left(E_n ,E_{\infty}\right)\to 0$ implies that 
    \begin{align*}
        \int_{\Omega} \ln\left(\frac{E_{m}}{E'}\right) \, \mathrm{d}P
        &\to \int_{\Omega} \ln\left(\frac{E_\infty}{E'}\right) \, \mathrm{d}P
    \intertext{so that} 
        \sup_E \int_{\Omega} \ln\left(\frac{E}{E'}\right) \, \mathrm{d}P
        &=\int_{\Omega} \ln\left(\frac{E_{\infty}}{E'}\right) \, \mathrm{d}P \, .
    \intertext{Hence}
        \sup_E \int_{\Omega} \ln\left(\frac{E}{E_{\infty}}\right) \, \mathrm{d}P
        &=0
    \end{align*}
    Therefore $E_{\infty}$ is a strongest \E-statistic.
    
    Assume that both $E_1$ and $E_2$ are strongest \E-variables. Then they are both stronger than the average $\bar{E}=\left(E_1+E_2\right)/2$. Hence
\begin{equation*}
    0\leq m_{\gamma}^2 \left(E_1,E_2\right)
    =\frac{1}{2}\int\left(\ln\left(\frac{\bar{E}}{E_1}\right)+\ln\left(\frac{\bar{E}}{E_2}\right)\right) \, \mathrm{d}P
    \leq 0.
\end{equation*}
Therefore $E_1=E_2$ $P$-almost surely.
\end{proof}
\begin{proof}[Proof of Theorem \ref{thm:LR_optim}]
Firstly, since $\hat E>0$ holds $P$-almost surely, we have that $\hat E$ is stronger than any $E'\in \mathcal{E}_\mathcal{C}$ with $P(E'=0)>0$. 

Secondly, let $E\in \mathcal{E}_\mathcal{C}$ be an \E-statistic for which $E>0$ holds $P$-almost surely. 
Furthermore, let $Q_n$ be a sequence of measures in $\mathcal{C}$ such that $D(P\|Q_n\rightsquigarrow \mathcal{C})\to 0$.
We can define a sequence of sub-probability measures $R_n$ by $R_n(F)=\int_F E \, \mathrm{d} Q_n$, which satisfies $\nicefrac{\mathrm{d}R_n}{\mathrm{d} Q_n} = E$.
We see
\begin{align*}
    \int_\Omega \ln\left(\frac{\hat E}{E}\right)\, \mathrm{d}P 
    &= \int_\Omega \ln\left(\frac{\mathrm{d}Q_n}{ \mathrm{d}\hat Q }  \right)\,\mathrm{d}P +D(P\|R_n)+(P(\Omega)-R_n(\Omega))\\
    &\geq \int_\Omega \ln\left(\frac{\mathrm{d}Q_n}{ \mathrm{d}\hat Q }  \right)\,\mathrm{d}P.
\end{align*}
The last expression goes to zero as $n\to \infty$, so we see that $\hat E$ is stronger than $E$.
\end{proof}

\begin{proof}[Proof of Proposition~\ref{prop:estat_implies_ripr}]
Using the fact that $\ln(x)\leq x-1$ for $x>0$, we see
    \[D(P\|Q^*\rightsquigarrow Q)
    =\int_\Omega \ln \frac{\mathrm{d}Q}{\mathrm{d}Q^*}\, \mathrm{d}P
    \leq \int_\Omega \left(\frac{\mathrm{d}Q}{\mathrm{d}Q^*}-1 \right)\, \mathrm{d}P
    = \int_\Omega \frac{\mathrm{d}P}{\mathrm{d}Q^*} \, \mathrm{d}Q -1
    \leq 0,\]
    where the last inequality follows from the fact that $\nicefrac{\mathrm{d}P}{\mathrm{d}Q^*}$ is an \E-statistic. 
\end{proof}

\begin{proof}[Proof of Theorem~\ref{thm:approx}]
Without loss of generality, assume that $\int_\Omega \nicefrac{q'}{q} \, \mathrm{d}P = 1 + \epsilon$ for some $\epsilon >0$.
For the sake of brevity, we write $c_{\beta}:=\|\nicefrac{q'}{q}\|_{1+\beta}^{1+\beta} $. 
We now define a function $g:[0,1]\to \mathbb{R}_{\geq 0}$ as
\[g(\alpha):=D\left(P\|(1-\alpha)Q+\alpha Q'\rightsquigarrow \mathcal{C}\right).\] 
Notice that $g(0)=\delta$ and $g(\alpha)\geq 0$, since $(1-\alpha)Q+\alpha Q'\in\mathcal{C}$.
This function and its derivatives will guide the rest of the proofs, and we now list some properties that we will need: 
\begin{align}
 g'(\alpha) &:= \frac{ \mathrm{d}}{\mathrm{d}\alpha} g (\alpha) = \int_\Omega \frac{q-q'}{(1-\alpha) q + \alpha q'}\, \mathrm{d}P,
 \intertext{so that}
g'(0) &= \int_\Omega \left(1 -\frac{q'}{q} \right)\, \mathrm{d}P = -\epsilon, \\
g''(\alpha) &:=\frac{\mathrm{d}^2}{\mathrm{d} \alpha^2} 
g (\alpha) = \int_\Omega \left(\frac{q'- q}{(1-\alpha) q + \alpha q'}\right)^2 \, \mathrm{d}P,  \label{eq:secondderivative} \\ 
\intertext{so that} g''(0) & = \int_\Omega \left(1 -\frac{q'}{q} \right)^2 \, \mathrm{d}P =  1 - 2(1+\epsilon) + c_1 \nonumber \intertext{and}
0 \leq g''(\alpha)&\leq \frac{1}{(1-\alpha)^2} g''(0). \label{eq:secondderbound}
\end{align}
We now prove (\ref{eq:rateofconv}).  
We start with the case $\beta =1$ and will use the result for $\beta=1$ to prove the case for $\beta < 1$. The proof for 
the case $\beta > 1$ comes later; it requires a completely different proof.

Case {\em $\beta = 1$}. The general idea is simple: at $\alpha=0$ the function $g(\alpha)$ is equal to $\delta$ and has derivative $-\epsilon$. Its second derivative is positive and bounded by constant times $g''(0)\leq c_1$ for all $\alpha \leq \nicefrac12$. Thus, if $\epsilon$ is larger then a certain threshold, $g(\alpha)$ will become negative at some $\alpha \leq \nicefrac12$, but this is not possible since $g$ is a description gain and we would arrive at a contradiction. The details to follow simply amount to calculating the threshold as a function of $\delta$.

% By Taylor's theorem, we have for any $\alpha \in [0,\nicefrac{1}{2}]$  that 
% \begin{align*}
% g'(\alpha) & \leq  g'(0)+ \alpha \max_{0 \leq \alpha^{\circ} \leq \alpha} g''(\alpha^{\circ}) \\
% & \leq  g'(0)+ \alpha 4 g''(0)  \leq - \epsilon +4 \alpha c,
% \end{align*} 
% where we use the properties derived above.
% It follows that $g'(\alpha) \leq -\nicefrac{\epsilon}{2}$ for all $\alpha \in [0,\alpha^*]$ where $\alpha^* = \min \{\nicefrac{1}{2}, \nicefrac{\epsilon}{8c}\}$.
% Consider the case that $\alpha^*=\nicefrac{\epsilon}{8c}<\nicefrac{1}{2}$. 
% By Taylor's theorem,  
% \[g(\alpha^*) \leq  g(0) + 
% \alpha^* \max_{\alpha^{\circ} \in [0,\alpha^*]} g'(\alpha^\circ) \leq \delta - \frac{\alpha^* \epsilon}{2} \leq \delta - \frac{\epsilon^2}{16c}.\] 
% Since $g(\alpha^*) \geq  0$ it follows that $\epsilon^2  \leq \delta 16 (1+c)$, so we can derive $\epsilon \leq 4 \left(\delta c\right)^{\nicefrac{1}{2}}$.
% Analogously, if $\alpha^* = \nicefrac{1}{2}$, i.e.\ $\nicefrac{\epsilon}{8(1+c)} > \nicefrac{1}{2}$, we can derive $\epsilon \leq 4 \delta$, and the result follows.

By Taylor's theorem, we have for any $\alpha \in [0,\nicefrac{1}{2}]$  that 
\begin{align*}
g(\alpha) &= g(0)+ g'(0)\alpha +  \max_{0 \leq \alpha^{\circ} \leq \alpha}\frac{g''(\alpha^{\circ})}{2} \alpha^2\\
&\leq  g(0)+  g'(0)\alpha + 2g''(0)\alpha^2 \\
&\leq \delta - \epsilon \alpha + 2\alpha^2 c_1,
\end{align*} 
where we use the properties derived above.
This final expression has a minimum in $\alpha^*=\min\{ \nicefrac{\epsilon}{4c_1},\nicefrac12 \}$. % Note: (1) second derivative is 4c, so it is a minimum  (2) if \nicefrac{\epsilon}{4c} > \nicefrac12, then 4\alpha c-\epsilon \leq 4\alpha c -2c\leq 0, so decreasing function => minimum in 1/2. 
By non-negativity of $g$, we know that $\delta-\epsilon\alpha^*+2\alpha^{*2}c_1\geq 0$.
This gives $\epsilon \leq (8c_1\delta)^{\nicefrac12}$ in the case that $\alpha^*=\nicefrac{\epsilon}{4c_1}<\nicefrac12$, and $\epsilon \leq 2\delta+c_1$ otherwise.
In the latter case, it holds that $c_1<\nicefrac{\epsilon}{2}$, so the bound can be loosened slightly to find the simplification $\epsilon \leq 4\delta$. 
This concludes the proof for $\beta=1$, which we now use to prove Case $\beta<1$.

Case {\em $\beta < 1 $}. For any $a>0$, it holds that
\begin{equation}\label{eq:markovtype}
    \int_\Omega \frac{q'}{q}\, \mathrm{d}P= \int_\Omega \frac{q'}{q} \mathbf{1}_{\{\nicefrac{q'}{q}\leq a\}}\, \mathrm{d}P+\int_\Omega \frac{q'}{q} \mathbf{1}_{\{\nicefrac{q'}{q}> a\}}\, \mathrm{d}P.
\end{equation}
We write $q'':=q' \mathbf{1}_{\{\nicefrac{q'}{q}\leq a\}}$ and we will bound the first term on the right-hand side of~\eqref{eq:markovtype} using the proof above with $Q'$ replaced by $Q''$. 
Since $Q''$ is not necessarily an element of $\mathcal{C}$, we need to verify non-negativity, which follows because
for each $\alpha\in (0,1)$, we have that $D(P\|(1-\alpha)Q+\alpha Q''\rightsquigarrow \mathcal{C})\geq D(P\|(1-\alpha)Q+\alpha Q'\rightsquigarrow \mathcal{C})\geq 0$.
Furthermore, it holds that
\begin{align*}
    \left\|\frac{q''}{q}\right\|_2^2&=\int_\Omega \left(\frac{q''}{q}\right)^2 \, \mathrm{d}P \\
    &= \int_\Omega \left(\frac{q''}{q}\right)^{1+\beta} \left(\frac{q''}{q}\right)^{1-\beta} \, \mathrm{d}P\\
    &\leq  a^{1-\beta} c_{\beta}
\end{align*} 
The results above therefore give 
\[\int_\Omega \frac{q''}{q}\, \mathrm{d}P\leq 1+ \max\{(8a^{1-\beta} c_{\beta} \delta)^{\nicefrac12}, 2\delta\}.\] 
For the second term on the right-hand side of~\eqref{eq:markovtype}, we use a Markov-type bound, i.e.
\begin{align*}
    \int_\Omega \frac{q'}{q} \mathbf{1}_{\{\nicefrac{q'}{q}> a\}}\, \mathrm{d}P 
    &\leq \int_\Omega \frac{q'}{q}\left( \frac{\nicefrac{q'}{q}}{a}\right)^\beta \mathbf{1}_{\{\nicefrac{q'}{q}> a\}}\, \mathrm{d}P \\
    &\leq a^{-\beta} c_{\beta}.
\end{align*}  
Putting this together gives
\[ \int_\Omega \frac{q'}{q}\, \mathrm{d}P\leq 1+ \max\{(8a^{1-\beta} c_{\beta}\delta)^{\nicefrac12}, 4\delta\}+ a^{-\beta} c_{\beta}.\]
Since this holds for any $a$, we now pick it to minimize this bound. 
To this end, consider 
\begin{align*}
    \frac{\mathrm{d}}{\mathrm{d}a} (8a^{1-\beta} c_{\beta}\delta)^{\nicefrac12}+a^{-\beta}c_{\beta}&=\frac{(1-\beta)(8 c_{\beta}\delta)^{\nicefrac12}}{2}a^{-\nicefrac{(1+\beta)}{2}} -\beta a^{-(1+\beta)}c_{\beta}.
\end{align*}
Setting this to zero, we find
\begin{align*}
    a^*=\left(\frac{\beta c_{\beta}^{\nicefrac12}}{(1-\beta)(2\delta)^{\nicefrac12}}\right)^\frac{2}{1+\beta}.
\end{align*}
The proof is concluded by noting that
\begin{align*}
    (8{a^*}^{1-\beta}c_\beta\delta)^{\nicefrac12}
    &=\left(8 \left(\frac{\beta c_{\beta}^{\nicefrac12}}{(1-\beta)(2\delta)^{\nicefrac12}}\right)^{2\frac{1-\beta}{1+\beta}} c_\beta\delta\right)^{\nicefrac12}\\
    &=2c_\beta^{\nicefrac{1}{(\beta+1)}}(2\delta)^{\nicefrac{\beta}{(\beta+1)}}\left(\frac{\beta}{1-\beta}\right)^{\frac{1-\beta}{1+\beta}}
\intertext{and} 
    {a^*}^{-\beta}c_\beta
    &=\left(\frac{\beta c_{\beta}^{\nicefrac12}}{(1-\beta)(2\delta)^{\nicefrac12}}\right)^{\frac{-2\beta}{1+\beta}} c_\beta\\
    &=c_\beta^{\nicefrac{1}{(\beta+1)}} \left(\frac{\beta}{1-\beta}\right)^{\frac{-2\beta}{1+\beta}}(2\delta)^{\nicefrac{\beta}{(1+\beta)}}.
\end{align*}

Case {\em $\beta > 1$}.
We now prove the result for $\beta\in (1,\infty)$; the proof for $\beta=\infty$ follows by a minor modification of (\ref{eq:gga}).
If $\epsilon \leq 0$ there is nothing to prove, so without loss of generality we can write $\epsilon = \gamma \delta$ for some $\gamma > 0$; we will bound $\gamma$. 
Whereas the previous proof exploited the fact that the second derivative $g''(\alpha)$ was bounded above in terms of $\delta$ and hence `not too large', the proof below uses the condition that $c_\beta$ is finite to show first, (a), that $g''(\alpha)$ can also be bounded {\em below\/} in terms of $(\gamma, \delta)$. 
Therefore, if $\epsilon$ exceeds a certain threshold, as $\alpha$ moves away from the $\alpha^*$ at which $g(\alpha)$ achieves its minimum in the direction of the furthest boundary point (i.e. if $\alpha^* < \nicefrac12$, we consider $\alpha \uparrow 1$, if $\alpha^* \geq \nicefrac12$ we consider $\alpha \downarrow 0$), $g(\alpha)$ will become larger than $K\delta$ or $\delta$ respectively, and we arrive at a contradiction. (b) below gives the detailed calculation of this threshold. 

{\em Proof of (a).\/}
Fix some $0 \leq \tilde{\alpha} < 1$ (we will derive a bound for any such $\tilde{\alpha}$ and later optimize for $\tilde{\alpha}$; for a sub-optimal yet easier derivation take $\tilde{\alpha}=\nicefrac12$). By Taylor's theorem,  we have 
$0 \leq g(\tilde{\alpha}) = \delta - \tilde{\alpha} \epsilon + (\nicefrac{1}{2}) \tilde{\alpha}^2 g''(\alpha^{\circ})$ for some $0 \leq \alpha^{\circ} \leq \tilde{\alpha}$. Plugging in  $\epsilon = \gamma\delta$ we find that
$$
g''(\alpha^{\circ}) \geq \frac{2}{\tilde{\alpha}^2} (\tilde{\alpha} \gamma -1) \delta.
$$
This gives a lower bound on $g''(\alpha^{\circ})$ for {\em some\/} $\alpha^{\circ}$ in terms of $(\gamma,\delta)$. We now turn this into a weaker lower bound on {\em all\/} $\alpha$. 
First,  using (\ref{eq:secondderbound}) and then $\alpha^{\circ} \leq \tilde{\alpha}$ and then the above lower bound,  we find
\begin{align}
    g''(0) & \geq \max_{\alpha\in [0,\tilde{\alpha}]} (1- \alpha)^2 g''(\alpha) \geq (1-\alpha^{\circ})^2 g''(\alpha^{\circ}) \nonumber  \\ \label{eq:ggg}
&    
\geq (1- \tilde{\alpha})^2 g''(\alpha^{\circ})
    \geq 2 f_{\tilde{\alpha}} (\gamma,\delta),
    \end{align}
where $f_{\tilde{\alpha}}(\gamma,\delta) := ((1-\tilde{\alpha})/\tilde{\alpha})^2 (\tilde{\alpha} \gamma -1) \delta$ is a function that is linear in $\gamma$ and $\delta$. 
We have now lower bounded $g''(0)$ in terms of $\gamma,\delta$. We next show that, under our condition that $c_\beta < \infty$, this implies a (weaker) lower bound on $g''(\alpha)$ for all $\alpha$. For this, 
    fix any $C > 1$. We have for all $0 < \alpha \leq 1$:
\begin{align}
    g''(\alpha)  & \geq \int_\Omega {\mathbf 1}_{q' \leq Cq} \cdot 
\left(\frac{q'- q}{(1-\alpha) q + \alpha q'}\right)^2 \, \mathrm{d} P  \nonumber \\
    & \geq \int_\Omega {\mathbf 1}_{q' \leq Cq} \cdot
    \left(\frac{q'- q}{(1-\alpha) q + \alpha Cq}\right)^2 
    \, \mathrm{d}P \nonumber \\
  & = \int_\Omega {\mathbf 1}_{q' \leq Cq} \cdot
    \left(\frac{q'- q}{q}\right)^2 
    \, \mathrm{d}P \cdot \frac{1}{(1+\alpha (C-1))^2}  \nonumber  \\
 & \geq  \frac{1}{(1+(C-1))^2}  \left( g''(0) - \int_\Omega {\mathbf 1}_{q' > Cq} \cdot
    \left(\frac{q'}{q} - 1 \right)^2 
    \, \mathrm{d}P  \right) \nonumber  \\ \label{eq:gga}
    & \geq   \frac{1}{C^2}  \left(  2 f_{\tilde\alpha}(\gamma,\delta)  - C^{1-\beta} c_\beta
    \right),
\end{align}
where in the fourth line we used the definition of $g''(0)$, and in the fifth line we used (\ref{eq:ggg}) and a Markov-type bound on the integral, i.e. we used that $
\int_\Omega {\mathbf 1}_{q' > Cq} \cdot
    \left({q'}/{q} - 1 \right)^2 
    \, \mathrm{d}P $ is bounded by 
$$
\int_\Omega {\mathbf 1}_{q' > Cq} \cdot
    \left(\frac{q'}{q} \right)^2 
    \, \mathrm{d}P \leq \int_\Omega \left( \frac{\nicefrac{q'}{q}}{C} \right)^{\beta-1}  \cdot
    \left(\frac{q'}{q} \right)^2 \mathrm{d}P = C^{1-\beta} c_\beta.
$$
By differentiation we can determine the $C$ that maximizes the bound (\ref{eq:gga}). 
This gives $C^{1- \beta} = f_{\tilde\alpha}(\gamma,\delta) (\nicefrac{4}{c_\beta(1+\beta)})$.
%((\gamma/2 - 1)  \delta )^{-1 /\beta} c_1^{(2+\beta)/\beta}$ 
and with this choice of $C$, (\ref{eq:gga})
becomes
\begin{equation}
    g''(\alpha) \geq 
f_{\tilde{\alpha}}(\gamma,\delta)^{(\beta+1)/(\beta-1)} c_\beta^{2/(1-\beta)} h(\beta) 
    %c_{1}^{-2 (2+\beta)/\beta} \left((\gamma /2 -1) \delta  \right)^{\frac{2+\beta}{\beta}}.
\end{equation}
where $h(\beta) = (4/(1+\beta))^{2/(\beta-1)} \cdot 2 (\beta-1)/(1+\beta)$.
We are now ready to continue to:

{\em Proof of (b)}.  Let $\alpha^* \in [0,1]$ be the point at which $g(\alpha)$ achieves its minimum. 
If $\alpha^* \leq \nicefrac12$, a second-order Taylor approximation of $g(1)$ around $\alpha^*$
gives that 
\begin{align*}
    K \delta &\geq g(1) \geq \frac12 (1-\alpha^*)^2 \min_{\alpha \in [\alpha^*,1]} g''(\alpha)\\
&\geq \frac{1}{8} f_{\tilde{\alpha}}(\gamma,\delta)^{(\beta+1)/(\beta-1)} c_\beta^{2/(1-\beta)} h(\beta), 
\end{align*}
so that after some manipulations 
\begin{equation}\label{eq:K3}
f_{\tilde\alpha}(\gamma,\delta)^{(1+ \beta)/(\beta-1)} \leq 
8 K' c_\beta^{2 /(\beta-1)} \cdot h(\beta)^{-1} \delta ,
\end{equation} 
with $K'= K$. If $\alpha^* > \nicefrac12$, we perform a completely analogous second-order Taylor approximation of $g(0)$ around $\alpha^*$,  which will then give (\ref{eq:K3}) again but with $K'$ replaced by $1$. We thus always have (\ref{eq:K3}) with $K'= \max \{K, 1 \}$.
Unpacking $f_{\tilde{\alpha}}$ in (\ref{eq:K3}) and rearranging gives:
$$
\gamma \leq \frac{\tilde\alpha}{(1-\tilde\alpha)^2 }
\cdot V
+ \frac{1}{\tilde\alpha} 
$$
with 
$$V= c_\beta^{2/(1+\beta)}  \cdot \left(\frac{8K'}{h(\beta)} \right)^{\frac{\beta-1}{1+ \beta}}  \delta^{\frac{-2}{1+\beta}}.$$
%\gamma \leq 2 \cdot \left( c_1^2(8K )^{\frac{\beta}{2+\beta}}  \delta^{- \frac{2}{2+\beta}}+1  \right)
We now pick the $\tilde\alpha$ that makes both terms on the right equal, so that the right-hand side becomes equal to $\nicefrac{2}{\tilde{\alpha}}$. This is the solution to the equation $(\nicefrac{\tilde\alpha}{(1-\tilde\alpha)})^2 V= 1$ which must clearly be obtained for some $0 < \tilde\alpha < 1$, so this $\tilde\alpha$ satisfies our assumptions. Basic calculation gives
$$
\gamma \leq \frac{2}{\tilde\alpha} = 2 \cdot \left(V^{\nicefrac12} + 1 \right)
$$
and unpacking $V$ we obtain  
\begin{equation*}%\label{eq:wevegotit}
\epsilon = \gamma \delta \leq  c^* \cdot \delta^{\frac{\beta}{1+\beta}} + 2 \delta. 
\end{equation*}
where 
$$
c^* = c_\beta^{1/(1+\beta)}  \cdot \left(\frac{8K'}{h(\beta)} \right)^{\frac{\beta-1}{2(1+ \beta)}} .
$$
Unpacking $h(\beta)$ gives the desired result.

\end{proof}

\section{RIPr Strict Sub-Probability Measure}\label{app:supp}
In this appendix, we discuss a general way to construct a measure $P$ and convex set of distributions $\mathcal{C}$ such that the reverse information projection of $P$ on $\mathcal{C}$ is a strict sub-probability measure.
For simplicity, we take $\Omega=\mathbb{N}$ and $\mathcal{F}=2^\mathbb{N}$, though the idea should easily translate to more general settings.

\begin{proposition}
    Let $g:\mathbb{N}\to \mathbb{R}_{>0}$ be a function, and let $\mathcal{C}$ denote the set of measures $\{Q: \sum_i g\left(i\right) q(i)\leq\nu\}$ for some $\nu >0$. Then for any $P$ that is not in $\mathcal{C}$ we have that $E\left(i\right)=\nicefrac{g\left(i\right)}{\nu} $ is the optimal \E-statistic.
\end{proposition}
\begin{proof}
    The extreme points in $\mathcal{C}$ are the measure with total mass 0
    and measures of the form $\frac{\nu}{g\left(i\right)} \delta_i$, 
    i.e.\ measures concentrated in single points. An \E-statistic $E$ must satisfy 
    \[
    \sum_j E\left(j\right)\frac{\nu}{g\left(i\right)}\delta_i \left(j\right)\leq 1
    \]
    or, equivalently, 
    $E\left(i\right)\frac{\nu}{g\left(i\right)}\leq 1$. 
    Hence $E\leq \nicefrac{g}{\nu}$ so the optimal \E-statistic is $\nicefrac{g}{\nu}$.
\end{proof}

Let $g:\mathbb{N}\to \mathbb{R}_{>0}$ be any function that satisfies 
\[\lim_{n\to \infty} g(n)=0.\]
Furthermore, let $P$ denote a probability measure on the natural numbers such that 
\[\sum_i \frac{p(i)}{g(i)}=c\]
for some $c \in \mathbb{R}_{>0}$.
For $\nu \in (0,1/c)$ and let $\mathcal{C}_{\nu}$ denote the set of measures $\{Q: \sum_i g\left(i\right) q(i)\leq\nu\}$.
Note that we do not yet require all measures in $\mathcal{C}_{\nu}$ to be probability measures so that the set $\mathcal{C}_{\nu}$ is compact.
It follows that there exists a unique element of $\mathcal{C}$ that minimizes $\sum_i p(i) \ln \left( p(i)/q(i)\right)$.

The optimal \E-statistic is $E_{\nu}=\nicefrac{g}{\nu}$, and we may define the measure $\hat Q_\nu$ by 
\[
\hat q_{\nu}(i)=\frac{p(i)}{E_\nu(i)}=\nu p(i)/g(i),
\]
and we can check that $\hat Q_{\nu} \in \mathcal{C}$. Hence $\hat Q_{\nu}$ minimizes $\sum_i p(i) \ln \left( p(i)/q(i)\right)$.

This is a strict sub-probability measure:
\begin{align*}
\sum_i \hat q_{\nu}(i)&=\nu \sum_i  \frac{p(i)}{g(i)}\\
&=\nu c\\
&< 1,
\end{align*}
where we use that $\nu<\nicefrac{1}{c}$.

The next step is to prove that the information projection does not change if we restrict to the set of probability measures in $\mathcal{C}_{\nu^*}$, which we denote by $\tilde{\mathcal{C}}_{\nu^*}$.
To this end, note first that for $\nu < \nu^*$, we have that $\sum g\left(i\right) q_\nu({i})<\nu^*$, so that for all $\nu < \nu^*$ there exists $n_\nu\in \mathbb{N}$ such that the probability measure defined by
\[
    q_\nu(i)+\left(1-\sum_j  q_\nu(j)\right)\delta_{n_\nu}\left(i\right)
\]
is an element of $\tilde{\mathcal{C}}_{\nu^*}$.
Hence 
\begin{align*}
D(P\|\tilde{C}) &\leq D\left(P\left\Vert  Q_{\lambda}+\left(1-\sum_j q_\nu({i})\right)\delta_{n_\nu} \right.\right)\\
&=\sum_{i\in \mathbb{N}}p(i)\ln\left(\frac{p(i)}{Q_\nu({i})+\left(1-\sum_{j\in \mathbb{N}} q_\nu(j)\right)\delta_{n_\nu}\left(i\right)}\right)\\
  &=-p({n_\nu})\ln\left(\frac{p(n_\nu)}{q_\nu({n_\nu})}\right)+p(n_\nu)\ln\left(\frac{p(n_\nu)}{q_\nu({n_\nu})+1-\sum_{j\in \mathbb{N}} q_\nu({j})}\right)
  +\sum_{i=1}^{\infty}p({i})\ln\left(\frac{p(i)}{q_\nu({i})}\right).
 \end{align*}
The first term can be written as 
\begin{align*}
p({n_\nu})\ln\left(\frac{p(n_\nu)}{q_\nu({n_\nu})}\right)
 & =q_\nu({n_\nu})\frac{p(n_\nu)}{q_\nu({n_\nu})}
 \ln\left(\frac{p(n_\nu)}{q_\nu({n_\nu})}\right)\\
 & =q_\nu({n_\nu})\frac{g\left(n_\nu\right)}{\nu}\ln\left(\frac{g\left(n_\nu\right)}{\nu}\right)
\end{align*}
Then notice that for $\nu \to \nu^*$, we must have that $n_\nu \to \infty$.
Using that $c\ln\left(c\right)\to0$
for $c\to 0$ we see the first term tends to 0
 for $\nu\to \nu^*$. 
Similarly, the second term can be written as 
\begin{multline}
p(n_\nu)\ln\left(\frac{p(n_\nu)}{q_\nu({n_\nu})+1-\sum_{j\in \mathbb{N}} q_\nu({j})}\right)= \\ 
\nonumber
  \left(q_{\nu}(n_\nu)+1-\sum_{j\in \mathbb{N}} q_\nu(j)\right)
  \frac{p(n_\nu)}{q_\nu({n_\nu})+1-\sum_{j\in \mathbb{N}} q_\nu({j})}
  \times\ln\left(\frac{p(n_\nu)}{q_\nu({n_\nu})+1-\sum_{j\in \mathbb{N}} q_\nu({j})}\right).
\end{multline}
We also have 
 \[
 \frac{p(n_\nu)}{q_{\nu}(n_\nu)+1-\sum q_\nu(i)}\to 0
 \]
for $\nu\to \nu^*$ and using that $c\ln\left(c\right)\to0$
for $c\to0$ we get the second term tends to 0 for $\nu\to \nu^*$. 
 Therefore we see
\begin{align*}
     D(P\|\tilde{\mathcal{C}})&\leq \lim_{\nu \to \nu^*}D\left(P\left\Vert Q_{\nu}+\left(1-\sum_i q_\nu({i})\right)\delta_{n_\nu} \right.\right)\\
     &\leq
     \sum_i p(i)\ln\left(\frac{p(i)}{q_{\nu^*}({i})}\right)\\
     &=\inf_{Q\in \mathcal{C}} \sum_i p(i)\ln \left(\frac{p(i)}{q(i)} \right).
\end{align*}
The inequality trivially also holds the other way around, so we find that 
\[D(P\|\tilde{\mathcal{C}})=\inf_{Q\in \mathcal{C}} \sum_i p(i)\ln \left(\frac{p(i)}{q(i)} \right).\] 
It follows that $Q_{\nu^*}$ is a strict sub-probability measure, and at the same time it is the reverse information projection of $P$ onto $\tilde{\mathcal{C}}_{\nu^*}$.

\section{Convexity}\label{ap:convexity}
One of the main assumptions made throughout the main text is that the set of measures $\mathcal{C}$ is convex, i.e. closed under finite mixtures.
However, one can also consider stronger notions of convexity, such as $\sigma$-convexity and Choquet-convexity.
In this appendix, we investigate whether considering different levels of convexity can change the reverse information projection.
\begin{definition}
    A set $\mathcal{C}'$ of measures is said to be $\sigma$\emph{-convex} if $Q_1\, ,Q_2\, ,Q_3\dots\in\mathcal{C}'$ implies that $\sum_{i=1}^\infty w_i Q_i \in \mathcal{C}'$ when $w_i\geq 0$ and $\sum_{i=1}^\infty w_i =1.$ 
    The $\sigma$-convex hull of a set of measures $\mathcal{C}$, denoted by $\sigma$-$\mathrm{conv}(\mathcal{C})$, is the smallest $\sigma$-convex set containing $\mathcal{C}$. 
\end{definition}
In order to avoid topological complications we will restrict the discussion of Choquet-convexity to Polish spaces, i.e. spaces for which there exists a complete metric that generates the topology. 
That is, assume that $\Omega$ is a Polish space equipped with the Borel $\sigma$-algebra. 
Let $\Theta$ be another Polish space and let $\{Q_{\theta}:\theta\in\Theta\}$ denote a parameterized set of probability measures on $\Omega$ such that $\theta\to\int_{\Omega}f\,\mathrm{d}Q_{\theta}$ is Borel measurable for any measurable function $f:\Omega\to \mathbb{R}$. 
Then for any probability measure $\nu$ on $\Theta$ the \emph{Choquet-convex mixture} $\mu_{\nu}$ can be defined by
\begin{equation*}
    \int_{\Omega}f\,\mathrm{d}\mu_{\nu}=\int_{\Theta}\left(\int_{\Omega}f\,\mathrm{d}\mu_{\theta}\right)\,\mathrm{d}\nu,
\end{equation*}
for any measurable function $f:\Omega\to \mathbb{R}$.
\begin{definition}
    A set $\mathcal{C}'$ of measures is said to be \emph{Choquet-convex} if it is closed under Choquet convex mixtures.
    The Choquet-convex hull of a set of measures $\mathcal{C}$ is the smallest Choquet-convex set that contains $\mathcal{C}$.
\end{definition}
So far, we have assumed that all of the measures in $\mathcal{C}$ are finite.
However, a countable or Choquet convex mixture of finite measures may not be finite. 
It follows that our results on the existence of the RIPr might not be applicable to the $\sigma$-convex and Choquet-convex hull of $\mathcal{C}$.
We therefore assume for the remainder of this section that all involved measures are sub-probability measures, in which case this issue does not arise.
With all of this in place, it is relatively straightforward to construct examples where the RIPr of $P$ on a convex set does not exist, while the RIPr of $P$ on its $\sigma$-convex hull does exist.
\begin{example}\label{ex:ripr_coincides}
    Let $P$ denote a geometric distribution on $\mathbb{N}_0$ and let $\mathcal{C}$ denote the set of probability measures on $\mathbb{N}_0$ with finite support. Then $D(P\|Q\rightsquigarrow \mathcal{C})=-\infty$ for any $Q\in\mathcal{C}$. Therefore the reverse information projection of $P$ on $\mathcal{C}$ is not defined according to the definitions given in this paper. 
    However, the $\sigma$-convex hull of $\mathcal{C}$ consists of all probability measures on $\mathbb{N}_0$, which implies that the reverse information projection on the $\sigma$-convex hull is well-defined and equals $P$.
\end{example}
\commentout{
%PETER I removed this since essentially the same phenomenon is shown (and also addressed/resolved!) by Example 4.4 in the Numeraire E-Value paper by Larsson et al. Also I feel it really does not belong at this place in the manuscript anyway...
\begin{example}
    Let $P$ denote the uniform distribution on the interval $[0,1]$ and let $Q$ denote 
    the uniform distribution on the interval $[0,2]$. 
    Furthermore, let $\mathcal{C}$ denote the set 
    of countable mixtures of $Q$ and discrete distributions on the interval $[0,2]$. We note that 
    $\mathcal{C}$ is $\sigma$-convex but it is not Choquet-convex. If $R$ is a discrete 
    distribution on $[0,2]$ then 
    $D(P\| (1-\alpha)Q+\alpha R)=\nicefrac{1}{(1-\alpha)}D(P\| Q),$ which is finite whenever 
    $0\leq\alpha<1.$ We see that $Q$ is the reverse information projection of $P$ on 
    $\mathcal{C}$ in the sense that $Q$ minimizes the divergence. The derivative is 
    \begin{equation*}
        \frac{\mathrm{d}P}{\mathrm{d}Q}=\begin{cases}
            2,\, \text{for}\, x\in [0,1];\\
            0,\, \text{else.}
        \end{cases}
    \end{equation*}
    We note that this derivative is not an \E-statistic on $\mathcal{C}.$ Our theorems 
    do not apply because the elements of $\mathcal{C}$ are not absolutely continuous 
    with respect to a common $\sigma$-finite measure as we have assumed in the 
    preliminaries. 
\end{example}
}
However, as the following results show, if the RIPr of $P$ on $\mathcal{C}$ does exist, then it must coincide with the RIPr of $P$ on $\sigma\text{-}\mathrm{conv}(\mathcal{C})$.
\begin{lemma}\label{lem:descr_gain_monotone}
    Let $P$ and $Q$  be sub-probability measures and let $Q_1 ,Q_2 ,\dots$ be a sequence of sub-probability measures such that $D(P\| Q \rightsquigarrow Q_1 )>-\infty$, and let $w_1 ,w_2 ,w_3 ,\dots$ be a sequence of positive numbers with sum 1. Then
    \begin{equation*}
        D\left(P\left\| Q \rightsquigarrow \frac{\sum_{i=1}^{n}w_i \cdot Q_i}{\sum_{i=1}^{n}w_i}\right.\right)\to
        D\left(P\left\| Q \rightsquigarrow {\sum_{i=1}^{\infty}w_i \cdot Q_i}\right.\right)
    \end{equation*}
    for $n\to\infty$.
\end{lemma}
\begin{proof}
    Firstly, note that  
    \[\ln \frac{\mathrm{d} \sum_{i=1}^{n+1} w_i Q_i}{ \mathrm{d}Q}\geq \ln \frac{\mathrm{d} \sum_{i=1}^n w_i Q_i}{ \mathrm{d}Q} \]
    and
    \[\int_\Omega \ln \frac{\mathrm{d} \sum_{i=1}^n w_i Q_i}{ \mathrm{d}Q} \, \mathrm{d} P \geq \int_\Omega \ln \frac{\mathrm{d}w_1 Q_1}{\mathrm{d}Q}\, \mathrm{d}P= D(P\|Q\rightsquigarrow Q_1)+\ln w_1 + (Q_1(\Omega)-Q(\Omega))>-\infty .\]
    
    Since $\sum_{i=1}^n w_iq_i \to \sum_{i=1}^\infty w_i q_i$ pointwise, applying the monotone convergence theorem to the sequence 
    \[
    \left(\ln \frac{\mathrm{d} \sum_{i=1}^n w_i Q_i}{\mathrm{d}Q}-\ln \frac{\mathrm{d}w_1 Q_1}{\mathrm{d}Q}\right)_{n\in \mathbb{N}}
    \]
    gives that
    \[\int_\Omega \ln \frac{\mathrm{d} \sum_{i=1}^n w_i Q_i}{\mathrm{d}Q}-\ln \frac{\mathrm{d}w_1 Q_1}{\mathrm{d}Q} \, \mathrm{d} P \to \int_\Omega \ln \frac{\mathrm{d} \sum_{i=1}^\infty w_i Q_i}{\mathrm{d}Q}-\ln \frac{\mathrm{d}w_1 Q_1}{\mathrm{d}Q}\, \mathrm{d} P. \]
    We get
    \[
    \int_\Omega \ln \frac{\mathrm{d} \sum_{i=1}^n w_i Q_i}{\mathrm{d}Q} \,\mathrm{d}P  \to \int_\Omega \ln \frac{\mathrm{d} \sum_{i=1}^\infty w_i Q_i}{\mathrm{d}Q}\, \mathrm{d} P
    \]
for $n\to\infty .$
    Finally, we see that
    \begin{align*}
        D\left(P\left\Vert Q \rightsquigarrow \frac{\sum_{i=1}^{n}w_{i}\cdot Q_{i}}{\sum_{i=1}^{n}w_{i}} \right. \right) & = \int_\Omega \ln \frac{\mathrm{d} \sum_{i=1}^n w_i Q_i}{\mathrm{d}Q} \, \mathrm{d} P - (Q_n(\Omega)- Q(\Omega)) - \ln \sum_{i=1}^n w_i \\
        &\to \int_\Omega \ln \frac{\mathrm{d} \sum_{i=1}^\infty w_i Q_i}{\mathrm{d}Q}\, \mathrm{d} P - (Q_\infty (\Omega)-Q(\Omega))\\
        &= D(P\| Q\rightsquigarrow Q_\infty),
    \end{align*}
    where $Q_\infty := \sum_{i=1}^\infty w_i Q_i$ and we use that $\ln\sum_{i=1}^n w_i\to 0$ and $Q_n(\Omega)\to Q_\infty(\Omega)$. 
    To see the latter, note that
    \[ Q_n(\Omega)= \int_\Omega \frac{\sum_{i=1}^n q_i(\omega) w_i}{\sum_{i=1}^n w_i} \, \mathrm{d}\mu(\omega), \]
    and $0\leq \nicefrac{\sum_{i=1}^n q_i(\omega) w_i}{\sum_{i=1}^n w_i}\leq \nicefrac{q_\infty(\omega)}{w_1}$, where the RHS integrates, so that the desired convergence follows from the dominated convergence theorem.
\end{proof}

\begin{theorem}\label{thm:ripr_convexity}
    Let $P$ be a finite measure and $\mathcal{C}$ a convex set of sub-probability measures such that $D(P\|Q\rightsquigarrow \mathcal{C})=0$.
    If $Q_1,Q_2,\dots$ is a sequence of measures in $\mathcal{C}$ such that $D(P\|Q_n\rightsquigarrow \mathcal{C})\to 0$, then $D(P\|Q_n\rightsquigarrow \sigma\text{-}\mathrm{conv}(\mathcal{C}))\to 0$.
\end{theorem}

\begin{proof}
    Fix $Q^*\in \mathcal{C}$ such that $D(P\|Q^*\rightsquigarrow \mathcal{C})\leq \varepsilon$ and let $\bar{Q} = \sum_{i=1}^\infty w_i Q_i \in \sigma\text{-}\mathrm{conv}(\mathcal{C})$ arbitrarily. 
    Let $s\in (0,1)$ and consider $\tilde{Q} := s\cdot Q^*+(1-s)\cdot \bar{Q} = \sum_{i=0}^\infty \tilde{w}_i Q_i$, where $Q_0:= Q^*$, $\tilde{w}_0=s$ and $\tilde{w}_i = (1-s)\cdot w_i$ for $i=1,2,\dots$.
    Note that $D(P\|Q^*\rightsquigarrow Q_0)=0$, so it follows from Lemma~\ref{lem:descr_gain_monotone} 
    that 
    \[ \lim_{n\to \infty} D\left(P\left\|Q^*\rightsquigarrow \frac{\sum_{i=0}^n \tilde{w}_i Q_i}{\sum_{i=0}^n \tilde{w}_i} \right. \right) = D(P\|Q^*\rightsquigarrow \tilde{Q}).\]
    The left hand side is, by definition of $Q^*$, bounded by $\varepsilon$ since $\nicefrac{\sum_{i=0}^n \tilde{w}_i Q_i}{\sum_{i=0}^n \tilde{w}_i}\in \mathcal{C}$, so that we find $D(P\|Q^*\rightsquigarrow \tilde{Q})\leq \varepsilon$.
    Furthermore, by concavity of the log,
    \[ \varepsilon \geq D(P\|Q^* \rightsquigarrow \tilde{Q}) \geq s\cdot D(P\|Q^*\rightsquigarrow Q_0)+(1-s)\cdot D(P\|Q^*\rightsquigarrow \bar{Q})= (1-s)\cdot D(P\|Q^*\rightsquigarrow \bar{Q}).\]
    Taking the limit of $s\to 0$, we see $D(P\|Q^*\rightsquigarrow \bar{Q})\leq \varepsilon$.
    Finally, the result follows by taking the supremum over $\bar{Q}$.
\end{proof}

We conjecture that if $\mathcal{C}$ is a $\sigma$-convex set of sub-probability measures and $\mathcal{C}'$ is the Choquet-convex hull of $\mathcal{C}$ then $D(P\|Q\rightsquigarrow \mathcal{C})=D(P\|Q\rightsquigarrow \mathcal{C}')$ for any sub-probability measures $P$ and $Q$ such that $P,Q$, and the sub-probability measures in $\mathcal{C}$ all have densities with respect to a common $\sigma$-finite measure.

\end{document}